\documentclass[10pt,journal]{IEEEtran}
\IEEEoverridecommandlockouts

\usepackage{graphicx}
\usepackage{amsmath,amsthm,amsfonts,amssymb}
\usepackage{cite,hyperref}
\usepackage{bm}
\usepackage{bbm}
\usepackage{url}
\usepackage{array}
\usepackage{color,soul}
\usepackage{multirow}
\usepackage{booktabs}
\usepackage[table,xcdraw]{xcolor}
\usepackage{enumitem}
\usepackage{subcaption}

\usepackage[english]{babel}
\usepackage{algorithm}
\usepackage{algorithmic}
\usepackage{setspace}
\usepackage[english]{babel}

\theoremstyle{plain}
\newtheorem{thm}{Theorem}
\newtheorem{lem}[thm]{Lemma}

\newtheorem{prop}[thm]{Proposition}

\newtheorem{rem}{Remark}
\newtheorem{sty1}{Theorem}
\newtheorem{defi}[sty1]{Definition}

\allowdisplaybreaks[4]

\begin{document}
\title{Directional Pinching-Antenna Systems}

\author{
Runxin Zhang,
Yulin Shao,
Yuanwei Liu

\thanks{R. Zhang is with the Department of Electronic Engineering, Tsinghua University, Beijing 100084, China.
}
\thanks{Y. Shao and Y. Liu are with the Department of Electrical and Electronic Engineering, The University of Hong Kong, Hong Kong S.A.R., China.
}
\thanks{
Correspondence: \url{ylshao@hku.hk}.}
}

\maketitle

\begin{abstract}
We propose a directional pinching-antenna system (DiPASS), a comprehensive framework that transitions PASS modeling from idealized abstraction to physical consistency. DiPASS introduces the first channel model that accurately captures the directional, pencil-like radiation of pinching antennas, incorporates a practical waveguide attenuation of $1.3$ dB/m, and accounts for stochastic line-of-sight blockage. A key enabler of DiPASS is our new ``equal quota division'' power allocation strategy, which guarantees predetermined coupling lengths independent of antenna positions, thereby overcoming a critical barrier to practical deployment. Our analysis yields foundational insights: we derive closed-form solutions for optimal antenna placement and orientation in single-PA scenarios, quantifying the core trade-off between waveguide and free-space losses. For multi-PA systems, we develop a scalable optimization framework that leverages directional sparsity, revealing that waveguide diversity surpasses antenna density in enhancing system capacity. Extensive simulations validate our analysis and demonstrate that DiPASS provides a realistic performance benchmark, fundamentally reshaping the understanding and design principles for future PASS-enabled 6G networks.
\end{abstract}


\section{Introduction}
The next generation of wireless communication systems, often envisioned as sixth-generation (6G) networks, is expected to deliver unprecedented performance in terms of data rate, reliability, and connectivity density \cite{dang2020should,liu2022evolution,shao2024theory}. Beyond merely scaling up from 5G, 6G aspires to create an intelligent, ubiquitous, and resilient wireless fabric capable of supporting immersive extended-reality (XR) services, tactile Internet, massive Internet of Things (IoT) connectivity, and AI-driven autonomous systems\cite{ali2025compact,wang2024electromagnetic,shao2021federated,kim2024role}. To satisfy these stringent requirements, future wireless infrastructures must simultaneously achieve ultra-high spectral and energy efficiency, robust coverage, and flexible adaptability in environments characterized by severe mobility, blockage, and multipath effects\cite{wang2024electromagnetic}.

A key technological trend underpinning this evolution is the exploitation of higher-frequency spectrum resources\cite{alsaedi2023spectrum,chow2024recent,ding2025flexible}. By pushing carrier frequencies from the sub-6 GHz band to millimeter-wave (mmWave), terahertz (THz), and even optical domains, researchers unlock orders of magnitude more bandwidth for wireless data transmission. However, high-frequency signals suffer from strong propagation loss, susceptibility to obstruction, and poor diffraction, which severely limit link reliability, especially in non-line-of-sight (NLoS) conditions\cite{ding2025flexible,liu2025pinching,wang2025modeling}. Consequently, the performance bottleneck of future wireless systems will not lie in the digital domain but rather in the physical deployment and reconfiguration of antennas capable of efficiently managing propagation in such hostile environments.

To confront these challenges, various reconfigurable and movable antenna architectures have been proposed to endow wireless systems with new physical degrees of freedom. Reconfigurable Intelligent Surfaces (RISs) can shape electromagnetic waves by dynamically adjusting passive reflecting elements \cite{tang2020wireless,dai2020reconfigurable,wu2021intelligent}; Fluid Antenna Systems (FASs) enable the antenna aperture to physically relocate within a confined space for diversity gain\cite{wong2020fluid}; and Movable Antennas (MAs) allow active elements to slide along tracks to improve spatial coverage\cite{zhu2023modeling,zhang2024polarization}. These technologies collectively illustrate an emerging paradigm shift: from static transceiver design to spatially adaptive, environment-aware communication architectures.

Within this landscape, the Pinching-Antenna System (PASS) has recently emerged as a particularly promising candidate for 6G mmWave and THz networks \cite{ding2025flexible,liu2025pinching,wang2025modeling}. In contrast to conventional fixed or reflective architectures, PASS utilizes low-cost dielectric waveguides as high-efficiency conduits for electromagnetic energy. Along these waveguides, one or multiple pinching antennas (PAs) can be flexibly activated at desired locations to form radiation points in close proximity to users. This design combines the high-frequency capability of waveguides with the spatial reconfigurability of movable antennas, enabling the radiating site to follow user movement or environmental variations. By effectively shortening the free-space propagation distance, PASS significantly mitigates path loss and line-of-sight (LoS) blockage, two dominant impairments in high-frequency communications.

The feasibility of PASS was first validated experimentally by NTT Docomo in 2022, where dielectric-waveguide prototypes demonstrated low-loss transmission and localized ``pinching'' radiation at $60$ GHz\cite{suzuki2022pinching}. Since then, PASS has rapidly evolved from a laboratory curiosity to a recognized research frontier. Numerous works have explored its potential from three complementary perspectives: system modeling, performance optimization, and theoretical analysis:
\begin{itemize}[leftmargin=0.45cm]
    \item First, system-modeling studies established analytical frameworks that divide the overall propagation path into in-waveguide and free-space components\cite{liu2025pinching,wang2025modeling,tegos2025minimum,xiao2025frequency}. These formulations allow the integration of dielectric-waveguide physics with mmWave channel models, paving the way for end-to-end capacity analysis.
    \item Second, optimization-oriented research exploited the new controllable degrees of freedom introduced by the PA's position\cite{ding2025flexible,chen2025dynamic,wang2025pinching}, beamforming\cite{wang2025joint}, and power allocation. Continuous or discrete algorithms have been proposed to maximize received power, sum rate, or fairness, while maintaining computational tractability\cite{wang2025antenna,wang2025pinching}.
    \item Third, performance-analysis works derived closed-form expressions for outage probability, average rate, and spectral efficiency, collectively showing that PASS can outperform traditional fixed-antenna systems and, in certain scenarios, even RIS-assisted networks in terms of spectral efficiency and coverage\cite{tyrovolas2025performance,ouyang2025capacity,liu2025pinching,samy2025pinching}.
\end{itemize}

Despite these advances, a pivotal simplification persists across the existing literature: the widespread assumption that a PA radiates as an omni-directional antenna. While this assumption undoubtedly simplifies mathematical modeling, it starkly contradicts the underlying electromagnetic principles. In reality, a PA functions as a leaky-wave antenna (LWA), generating a highly directional, pencil-like beam concentrated in the forward half-space. Overlooking this fundamental anisotropy leads to a significant overestimation of the effective radiated power and a distorted representation of interference patterns and true coverage zones.

The implications of this ``directional oversight'' are further amplified when considered alongside other practical constraints. The directional radiation pattern directly interacts with realistic waveguide attenuation, which can be as high as $1.3$ dB/m, and the prevalent impact of LoS blockage at high frequencies. An omni-directional model fails to capture the critical interplay between antenna orientation, energy decay along the waveguide, and the probabilistic nature of the link, leading to potentially misleading system design guidelines.

Motivated by this critical gap, this paper undertakes a fundamental re-examination of PASS by placing the directional nature of its radiation at the forefront. We introduce and analyze the directional pinching-antenna system (DiPASS), establishing a comprehensive framework that transitions from idealized abstraction to physically consistent modeling and optimization. Our main contributions are summarized as follows:
\begin{enumerate}[leftmargin=0.5cm]
    \item We establish the first physically consistent directional channel model for PASS. This model explicitly integrates the directional PA radiation via a Gaussian beam pattern, serving as the foundational element. Built upon this core, we incorporate a practical waveguide attenuation coefficient and stochastic LoS blockage, thereby creating a unified framework for system analysis. Crucially, to overcome the inflexibility of existing multi-PA power allocation schemes, we propose a novel ``equal quota division'' strategy. This scheme leverages the memoryless property of exponential decay, ensuring predetermined coupling lengths regardless of PA positions, thereby resolving a critical barrier to practical PASS deployment.
    \item Leveraging the directional model, we uncover the fundamental trade-off governing PASS performance: the competition between in-waveguide attenuation and free-space propagation reliability. For single-PA scenarios, we derive closed-form expressions for the optimal PA placement and orientation under directional radiation constraints. More importantly, we establish a necessary condition for the existence of the optimal PA position. These analytical results provide rigorous mathematical guidelines for system design, shifting the paradigm from empirical rules to theory-driven deployment.
    \item For DiPASS with multiple waveguides and PAs, we develop a scalable optimization framework. The highly directional beams naturally transform the complex joint design into a structured PA-user assignment problem, for which we design an efficient hybrid Hungarian-greedy algorithm. This approach enables a systematic comparison of resource allocation strategies, leading to two key design principles: i) waveguide diversity is systematically more effective than antenna density for boosting the system sum rate; ii) a clear performance-coverage trade-off exists between beamforming strategies, where weighted minimum mean-square error (WMMSE) beamformer maximizes sum-rate, while zero forcing (ZF) beamformer ensures broader user coverage.
\end{enumerate}

\section{System Model}\label{sec:II}
\subsection{Spatial Configuration}
We consider a typical pinching-antenna system in which a base station (BS) is connected to $N$ dielectric waveguides, as illustrated in Fig.~\ref{fig:system_model}. The waveguides extend along the $y$-axis and are uniformly deployed along the $x$-axis. The physical length of each waveguide is $D_y$, and the dielectric material introduces an internal attenuation coefficient of $\alpha_{\text{W}}$~dB/m. Along each waveguide, $L$ PAs are installed. Both the waveguides and PAs share the same cross-sectional dimensions of $a\lambda \times b\lambda$, where $\lambda$ is the wavelength.

\begin{figure}
    \centering
    \includegraphics[width=0.8\linewidth]{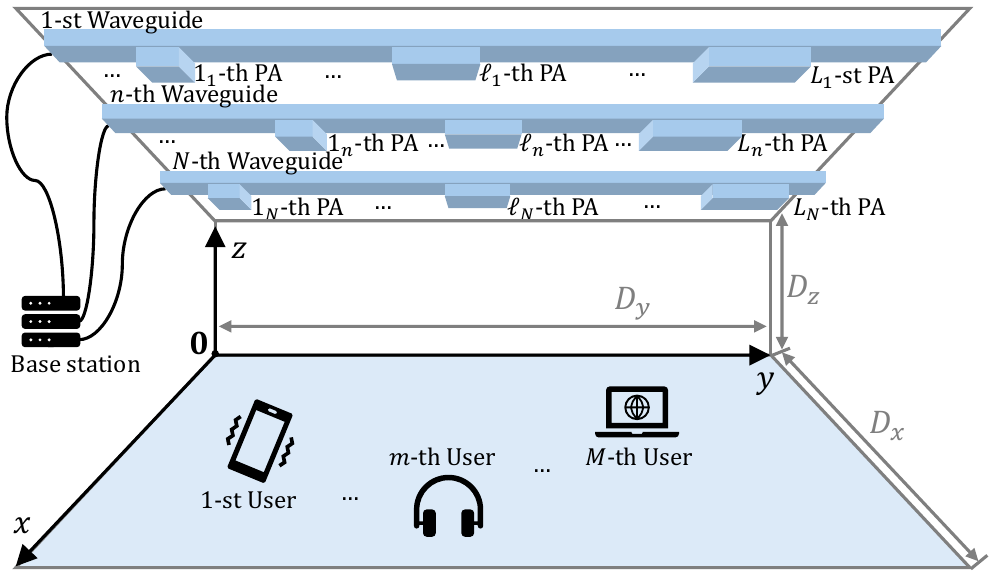}
    \caption{Architecture of a typical PASS. The BS employs $N$ dielectric waveguides hosting a total of $NL$ PAs to serve $M$ single-antenna users.}
    \label{fig:system_model}
\end{figure}

The area served by the PASS infrastructure is modeled as a three-dimensional (3D) box with dimensions $D_x \times D_y \times D_z$. Inside this region reside $M$ single-antenna users randomly distributed in space. The coordinate tuple $\bm{\psi}=(x,y,z)$ is used to describe the positions of waveguides, PAs, and users:
\begin{itemize}[leftmargin=0.5cm]
    \item {\it Waveguide (\text{W}):} We represent the position of a waveguide by the coordinate of the waveguide entrance: $\bm{\psi}_n^{(\text{W})} = (x_n^{(\text{W})}, y_n^{(\text{W})}, D_z)$, where $n = 1, \dots, N$. Since the waveguides extend along the $y$-axis and are uniformly distributed along the $x$-axis, we have $y_n^{(\text{W})} = 0$ and $x_n^{(\text{W})} = \frac{2n-1}{2N} D_x$.
    \item {\it PA (\text{P}):} For the $n$-th waveguide, we represent the position of the $\ell_n$-th PA by the point at which the PA radiates the signal, with coordinates given by $\bm{\psi}_{\ell_n}^{(\text{P})} = (x_{\ell_n}^{(\text{P})}, y_{\ell_n}^{(\text{P})}, z_{\ell_n}^{(\text{P})})$, where $\ell_n = (n-1)L + \ell$, $\ell = 1, \dots, L$, and $n = 1, \dots, N$. Since the PAs are mounted on waveguides, their $x$- and $z$-coordinates should be identical, hence we have $x_{\ell_n}^{(\text{P})} = x_n^{(\text{W})}$ and $z_{\ell_n}^{(\text{P})} = D_z$.  
    \item {\it User (\text{U}):} The coordinates of the $m$-th user are expressed as $\bm{\psi}_m^{(\text{U})} = (x_m^{(\text{U})}, y_m^{(\text{U})}, 0)$, where $m = 1, \dots, M$. Users are uniformly distributed. 
\end{itemize}

\subsection{Signal Transmission Model}
The signal transmission process begins at the BS, where information intended for the $M$ users is denoted by the vector $\bm{s} = \{ s_m \} \in \mathbb{C}^{M \times 1}$, with $\mathbb{E}\{ s_m \} = 1$. To manage multi-user interference, the BS employs a precoding matrix $\bm{W} = [ \bm{w}_1,\bm{w}_2, \dots, \bm{w}_M ] = \{ w_{n,m} \} \!\in\! \mathbb{C}^{N \times M}$  with  $\text{tr}\left( \bm{W}^H \bm{W} \right) = 1$.

The composite received signal at all users can be expressed as
\begin{equation}
    \bm{\xi} = \sqrt{P} \bm{H} \bm{W} \bm{s} + \bm{\nu},
\end{equation}
where 
$ P $ denotes the total transmit power and $ \bm{\nu} $ represents additive white Gaussian noise (AWGN). The overall channel matrix $ \bm{H} \in \mathbb{C}^{M \times N} $ captures the complete transmission path from the BS to the users.

To elucidate the unique signal propagation characteristics in PASS, we decompose $\bm{H}$ into three physically meaningful components
\begin{equation}
\bm{H} = \bm{H}^{(\text{P} \to \text{U})} \bm{H}^{(\text{W} \to \text{P})} \bm{\Lambda}.
\end{equation}
In the decomposition, 
\begin{itemize}[leftmargin=0.5cm]
    \item $\bm{\Lambda}$ is a block-diagonal matrix composed of $N$ all-ones vectors $\mathbf{1}^{(L \times 1)}$, which distributes the beamforming coefficients from $\bm{W}$ across the $L$ antennas on each waveguide.
    \item $\bm{H}^{(\text{W} \to \text{P})}$ is a diagonal matrix of size $NL \times NL$, characterizing the channel from the waveguide input to the PA radiation points, accounting for waveguide attenuation and phase variations.
    \item $\bm{H}^{(\text{P} \to \text{U})}$, with dimensions $M \times NL$, encapsulates the channel coefficients between the PAs and users, incorporating radiation patterns and propagation effects.
\end{itemize}

Focusing on individual user reception, the signal at the $m$-th user can be decomposed as
\begin{eqnarray}\label{e:received_sigal}
    \xi_m &=& \sqrt{P} \sum_{n=1}^{N} [\bm{H}]_{m,n} \sum_{i=1}^M [\bm{W}]_{n,i} s_i + \nu_m \notag\\
    &\triangleq& S_m + I_m + \nu_m,
\end{eqnarray}
where $[\star]_{i,j}$ denote the $(i,j)$-th element of a matrix and $\nu_m \sim \mathcal{N}(0, \sigma_m^2)$. Here, $S_m$ represents the desired signal component, while $I_m$ denotes the interference from other users:
\begin{eqnarray}
    S_m &=& \sqrt{P} \sum_{n=1}^{N} [\bm{H}]_{m,n} [\bm{W}]_{n,m} s_m, \\
    I_m &=& \sqrt{P} \sum_{n=1}^{N} [\bm{H}]_{m,n} \sum_{i=1,i\neq m}^M [\bm{W}]_{n,i} s_i.
\end{eqnarray}

While the framework above outlines a general PASS model, achieving a design that is physically consistent and practically relevant requires carefully integrating several key effects.

First, this work explicitly models the directional radiation characteristics of PAs. As stated in the Introduction, a PA operates as a LWA, emitting a pencil-like beam concentrated in the forward half-space (see Fig.~\ref{f:directional}). Nevertheless, numerous system-level analyses persist in using the omni-directional assumption for mathematical tractability. This oversight leads to inflated performance predictions and misrepresents coverage patterns. Our DiPASS framework is the first to systematically integrate a physically-accurate directional radiation pattern into the PASS channel model.

Moreover, to enhance the DiPASS's practicality, we adopt and integrate more realistic parameters for two critical effects:
\begin{itemize}[leftmargin=0.5cm]
    \item \textit{Waveguide attenuation}: Unlike existing works that either ignore waveguide attenuation \cite{chen2025dynamic,tegos2025minimum,wang2025pinching} or set a small coefficient (e.g., $\alpha_{W} = 0.08$ dB/m) \cite{ouyang2025rate,ding2025flexible,samy2025pinching}, this paper adopts the practical attenuation coefficient $\alpha_{W} = 1.3$ dB/m demonstrated in Docomo's experimental prototypes \cite{suzuki2022pinching}, thereby advancing PASS modeling closer to real-world implementation. This also allows us to analyze the non-trivial trade-off between in-waveguide and free-space propagation losses in Section \ref{sec:IV}.
    \item \textit{Stochastic LoS blockage}: For mmWave communications, LoS blockage is a dominant factor affecting link reliability. To capture this, we introduce a stochastic blockage model where the environment includes independent obstacles that affect the probability of LoS connectivity. This probability is characterized by the existence coefficient $\alpha_{\text{L}}$ for the LoS path between PAs and users.
\end{itemize}

In the following section, we present detailed formulations for $\bm{H}^{(\text{W} \to \text{P})}$ and $\bm{H}^{(\text{P} \to \text{U})}$ that incorporate these physical considerations.

\section{DiPASS Channel Model}\label{sec:III}
This section elaborates the detailed channel modeling for DiPASS that addresses the limitations identified above. We separately characterize the channel from the waveguide input to the PA radiation point, accounting for attenuation, phase variations, and coupling efficiency, followed by the channel from the PA radiation point to the user, employing a realistic directional propagation model.

\subsection{Waveguide-to-PA Channel}
To accurately model the signal journey from the waveguide input to the PA radiation point, we first incorporate exponential attenuation explicitly in the channel model and propose a new, practical power allocation strategy among PAs that overcomes the inflexibility of existing methods.

When multiple PAs are connected to the same waveguide, the distribution of power among them is determined by two key parameters\cite{liu2025pinching,okamoto2021fundamentals}:
\begin{itemize}
    \item \emph{Coupling Coefficient} $\kappa$: the fraction of guided power extracted per unit interaction distance, mainly determined by the waveguide material.
    \item \emph{Coupling Length} $\tau_{\ell_n}$: the physical interval over which the $\ell_n$-th PA remains adjacent to the $\ell$-th waveguide.
\end{itemize}

Prior works, which often neglected waveguide attenuation, typically adopted either \textit{equal power division} or \textit{proportional power division} \cite{liu2025pinching}. However, under the practical condition of exponential signal decay, both schemes require optimizing each PA's length based on its specific position to achieve the desired power split. This leads to a complex and inflexible design where PA lengths vary inconsistently along the waveguide, making real-world implementation impractical.

To resolve this fundamental issue of position-dependent and inflexible design, we propose a more practical scheme named equal quota division.

\begin{defi}[Equal Quota Division]
For a waveguide feeding $L$ PAs, this scheme allocates to each PA an equal power quota of $P/L$ at the waveguide input. The actual power received by a PA is its initial quota attenuated by the exponential decay from the input to its position.
\end{defi}

The critical advantage of this scheme stems from the memoryless property of exponential decay. Since the attenuation at any point depends only on the traversed distance, the amount of power remaining in the waveguide at any PA location is independent of how much power was coupled out by the preceding PAs. Consequently, the coupling length $\tau_{\ell_n}$ required for each PA to extract its $1/L$ share is determined solely by its \emph{sequential order} along the waveguide, rather than by its absolute position. This results in a fixed, predetermined set of coupling lengths $\{\tau_{\ell_n}\}$, making the PASS architecture highly flexible and practical for deployment. 

Building on this scheme, we now formalize the channel model from the waveguide input to the PAs.

\begin{prop}[Waveguide-to-PA Channel]\label{lem:H_WP}
Denote by $\tau_{\ell_n}$ the coupling length of the $\ell$-th PA on the $n$-th waveguide.
\begin{enumerate}[leftmargin=0.5cm]
    \item Under equal quota division, $\tau_{\ell_n}$ depends only on the index $\ell$ (i.e., the PA's order along the waveguide) and is given by
    \begin{equation}\label{e:couple_length}
        \tau_{\ell_n} = \frac{1}{\kappa}\arcsin\sqrt{\frac{1}{L+1-\ell}},~\forall n,
    \end{equation}
    where $\kappa$ is coupling coefficient. For notational clarity, we omit the subscript $n$ and write $\tau_{\ell_n}$ as $\tau_{\ell}$ hereafter.
    \item The channel matrix $\bm{H}^{(\text{W} \to \text{P})}$ is diagonal. Its $\ell_n$-th element, $h_{\ell_n}^{(\text{W} \to \text{P})}$, denoting the channel coefficient from the input of the $n$-th waveguide to its $\ell$-th PA, satisfies
\begin{equation}\label{e:H_WP}
    h_{\ell_n}^{(\text{W} \to \text{P})} = \sqrt{\frac{1}{L}} e^{-\frac{\alpha_{\text{W}}}{2} y_{\ell_n}^{(\text{P})}} e^{-j\frac{2\pi}{\lambda_g}y_{\ell_n}^{(\text{P})}}.
\end{equation}
\end{enumerate}
\end{prop}

\begin{proof}
We begin by expressing the channel coefficient as
\begin{align}\label{e:H_WP_detail}
h_{\ell_n}^{(\text{W} \to \text{P})} = &
\prod_{i=1}^{\ell-1}
\underbrace{e^{-\frac{\alpha_{\text{W}}}{2} \left(y_{i_n}^{(\text{P})} - y_{(i-1)_n}^{(\text{P})}\right)}}_{\text{(a)}}
\underbrace{\sqrt{1 - \underbrace{\sin^2(\kappa \tau_i)}_{\text{(b)}}}}_{\text{(c)}} \notag\\
\cdot & e^{-\frac{\alpha_{\text{W}}}{2} \left(y_{\ell_n}^{(\text{P})} - y_{(\ell-1)_n}^{(\text{P})} \right)}
\sin(\kappa \tau_\ell)
\underbrace{e^{-j\frac{2\pi}{\lambda_2}y_{\ell_n}^{(\text{P})}}}_{\text{(d)}}  \notag\\
=& \left[\prod_{i=1}^{\ell-1}\sqrt{1 - \sin^2(\kappa \tau_i)}
\sin(\kappa \tau_\ell)\right] \notag\\
&\cdot  e^{-\frac{\alpha_{\text{W}}}{2} y_{\ell_n}^{(\text{P})}} e^{-j\frac{2\pi}{\lambda_2}y_{\ell_n}^{(\text{P})}},
\end{align}
where (a) represents amplitude attenuation from the $(i-1)$-th PA to the $i$-th PA, and without loss of generality, we set $y_{0_n}^{(\text{P})} = 0$;
(b) denotes the proportion of power coupled out by the $i$-th antenna;
(c) is the amplitude corresponding to the proportion of power remaining in the waveguide;
(d) accounts for the phase shift introduced after traveling a distance of $y_{\ell_n}^{(\text{P})}$.

The three reorganized multiplicative terms in \eqref{e:H_WP_detail} correspond to equivalent coupling efficiency, waveguide attenuation, and phase shift inside the waveguide, respectively.
For equal quota division, each antenna receives a power share of $1/L$. Therefore, for the equivalent coupling efficiency, we have
\begin{equation}\label{e:coupling_efficiency}
    \prod_{i=1}^{\ell-1}\sqrt{1 - \sin^2(\kappa \tau_{i_n})}
\sin(\kappa \tau_{\ell_n}) = \sqrt{\frac{1}{L}}.
\end{equation}
Substituting \eqref{e:coupling_efficiency} into \eqref{e:H_WP_detail} yields \eqref{e:H_WP}.

Then, we prove \eqref{e:couple_length} by mathematical induction.
For $\ell = 1$, the product term is empty, and the condition reduces to
\begin{equation}
    \sin(\kappa \tau_1) = \sqrt{\frac{1}{L}}.
\end{equation}
Choosing $\tau_1 = \frac{1}{\kappa}\arcsin\sqrt{\frac{1}{L}}$ satisfies this condition.

Assume the condition holds for $\ell = j$, i.e.,
\begin{equation}
    \left(\prod_{i=1}^{j-1}\sqrt{1 - \sin^2(\kappa \tau_i)} \right)\sin(\kappa \tau_j) = \sqrt{\frac{1}{L}}.
\end{equation}
From the inductive hypothesis and the coupling length design, we have
\begin{equation*}
    \sqrt{1 - \sin^2(\kappa \tau_i)} = \sqrt{\frac{L-i}{L+1-i}}.
\end{equation*}

Now consider the case for $\ell = j+1$ with $\tau_{j+1} = \frac{1}{\kappa}\arcsin\sqrt{\frac{1}{L-j}}$:
\begin{eqnarray*}
    \hspace{-1cm} & \left(\prod_{i=1}^{j}\sqrt{1 - \sin^2(\kappa \tau_i)}\right)
    \sin(\kappa \tau_{j+1}) \\
    \hspace{-1cm}&\hspace{2.5cm}= \sqrt{\frac{L-1}{L} \cdot \frac{L-2}{L-1} \cdots \frac{L-j}{L-j+1}} \cdot \sqrt{\frac{1}{L-j}}
    = \sqrt{\frac{1}{L}}.\notag
\end{eqnarray*}

By mathematical induction, the proposed coupling lengths $\tau_\ell = \frac{1}{\kappa}\arcsin\sqrt{\frac{1}{L+1-\ell}}$, $\ell = 1, \dots, L$, satisfy the equal quota division condition for all $\ell$.
\end{proof}

\begin{rem}
   In the ideal case of zero waveguide attenuation ($\alpha_{\text{W}} = 0$), the proposed equal quota division scheme coincides with the conventional equal power division. Under this condition, the model in Proposition~\ref{lem:H_WP} ensures that each PA radiates identical signal power, thereby recovering a common assumption in prior work which neglected attenuation.
\end{rem}

\subsection{PA-to-User Channel}\label{sec:IIIB}
In real-world deployments, a PA functions as a LWA, exhibiting strongly directional radiation. This fact has often been overlooked by the omnidirectional assumption in system-level analyses for mathematical convenience. In a PASS, the radiation is typically achieved by introducing a specific structure (e.g., a ``pinching'' mechanism that perturbs the waveguide) which causes the guided wave to gradually leak energy into free space \cite{wang2025modeling}.
This physical mechanism inherently produces a pencil-like beam concentrated in the forward half-space, a classical result of waveguide antenna theory \cite{jackson2012leaky}.

Furthermore, at mmWave frequencies, LoS paths are highly susceptible to blockage by obstacles. Our channel model in this section explicitly incorporates both a physically accurate directional radiation pattern and a stochastic blockage model to capture these essential effects.

\subsubsection{Directional Radiation Pattern of a PA}
To model the directional radiation of a PA, we adopt the foundational Gaussian beam approximation for the fundamental mode of a rectangular dielectric waveguide \cite{okamoto2021fundamentals}. 
This model is well-established for describing the beam-like radiation from such structures. The electric field pattern, which characterizes the directional gain, can be described in the PA's local coordinate frame. When radiated from the origin and along the positive $y$-axis, the radiation pattern is given by (2.109) in \cite{okamoto2021fundamentals}:
\begin{eqnarray}\label{e:Upsilon}
&& \hspace{-0.8cm} \Upsilon\left( x, y, z \right) = \sqrt{\frac{w_1 w_2}{W_1 W_2}} B \exp \left\{ - \left[ \frac{x^2 }{W_1^2} + \frac{ z^2 }{W_2^2} \right] \right. \\
&& \hspace{0.8cm} \left. - j k n \left[ \frac{ x^2 }{2 R_1} + \frac{ z^2 }{2 R_2} + y \right] + j \frac{\Theta_1 + \Theta_2}{2} \right\}, \notag
\end{eqnarray}
where the parameters of the Gaussian beam $(W_1,W_2)$, $(R_1,R_2)$, and $(\Theta_1,\Theta_2)$ represent the beam radii, radii of curvature, and Gouy phases along the x- and z-axes, respectively:
\begin{itemize}[leftmargin=0.5cm]
    \item The beam radii $W_i = \lambda y / (\pi n w_i)$, $i=1,2$, where $n$ is the refractive index. The initial Gaussian widths satisfy $w_1/(a\lambda) = w_2/(b\lambda) = v$, where $v \approx 1.1$ is an empirical correction factor linked to the normalized waveguide frequency \cite{yeh2008essence}.
    \item Under the far-field approximation, the radii of curvature are taken as $R_1 = R_2 = y$.
    \item The Gouy phases are $\Theta_i = \tan^{-1}\big( \lambda y / (\pi n w_i) \big)$, $i=1,2$. 
    \item The normalization constant $B$ ensures $\iint_{xOz} \allowbreak |f(x,y,z)|^2 \allowbreak dx dz = 1$, yielding $B^2 = 2/(\pi w_1 w_2)$.
\end{itemize}

To utilize this model, we must first determine the user's position within the PA's local coordinate frame, with the PA's boresight aligned to the radiation beam. We define the antenna's orientation by an elevation angle $\theta_{\ell_n} \in (\pi/2, \pi]$ (i.e., the angle between the PA's radiation direction and the positive $z$-axis) and an azimuth angle $\varphi_{\ell_n} \in (-\pi, \pi]$ (i.e., the angle between the projection of this direction onto the $x$–$y$ plane and the positive $x$-axis). Note that $\theta_{\ell_n} = \pi$ corresponds to the antenna pointing vertically downward. The following lemma defines the required coordinate transformation.


\begin{lem}[Coordinate Transformation] \label{lem:psiU_tilde}
The coordinates of the $m$-th user in the local reference frame of the $\ell_n$-th PA, denoted as ${\tilde{\bm{\psi}}_{m,\ell_n}^{(\text{U})}} = \left( {\tilde{x}_{m,\ell_n}^{(\text{U})}}, {\tilde{y}_{m,\ell_n}^{(\text{U})}}, {\tilde{z}_{m,\ell_n}^{(\text{U})}} \right)$, are given by
\begin{equation}\label{e:psiU}
\hspace{-0.2cm}{\tilde{\bm{\psi}}_{m,\ell_n}^{(\text{U})}} = \bm{R}_x \left( \theta_{\ell_n}  \!-\! \frac{\pi}{2} \right) \bm{R}_z \left( \frac{\pi}{2} \!-\! \varphi_{\ell_n}   \right) \left( \bm{\psi}_m^{(\text{U})} \!-\! \bm{\psi}_{\ell_n}^{(\text{P})}  \right),
\end{equation}
where
$$\bm{R}_z(x) \!=\!\! \begin{bmatrix} \cos{x} & \!\!\!-\sin{x} & \!\!\!0 \\ \sin{x} & \!\!\!\cos{x} & \!\!\!0 \\ 0 & \!\!\!0 & \!\!\!1 \end{bmatrix} \text{ and }~ \bm{R}_x(x) \!=\!\! \begin{bmatrix} 1 & \!\!\!0 & \!\!\!0 \\ 0 & \!\!\!\cos{x} & \!\!\!-\sin{x} \\ 0 & \!\!\!\sin{x} & \!\!\!\cos{x} \end{bmatrix}$$ 
are the rotation matrices around the z-axis and x-axis, respectively.
\end{lem}

\begin{proof}
The PA's directional vector $\bm{n}_{\ell_n} = (\sin\theta_{\ell_n} \allowbreak \cos\varphi_{\ell_n}, \allowbreak \sin\theta_{\ell_n} \allowbreak \sin\varphi_{\ell_n}, \allowbreak  \cos \theta_{\ell_n})$ can be obtained by rotating the reference vector $(0,1,0)$ via two successive rotations: first by a rotation of $(\frac{\pi}{2} - \theta_{\ell_n}) $ around the $x$-axis in the right-hand screw direction, then by a rotation of $(\varphi_{\ell_n}  - \frac{\pi}{2})$ around the $z$-axis in the right-hand screw direction. This is expressed as
\begin{equation}
\bm{n}_{\ell_n}  \left( \varphi_{\ell_n} , \theta_{\ell_n}  \right)  = \bm{R}_z  \left( \varphi_{\ell_n}  - \frac{\pi}{2} \right) \bm{R}_x \left( \frac{\pi}{2} - \theta_{\ell_n}  \right) \begin{bmatrix} 0 \\ 1 \\ 0 \end{bmatrix}.
\end{equation}
Then, for any point in space, the transformation from global coordinates to local reference frame (i.e., relative coordinates) can be seen as an inverse operation, i.e., first a rotation of $-\left( \varphi_{\ell_n}  - \frac{\pi}{2} \right)$ around the $z$-axis, followed by a rotation of $-\left( \frac{\pi}{2} - \theta_{\ell_n}  \right)$ around the $x$-axis. This is given by \eqref{e:psiU}.
\end{proof}

\subsubsection{Integrated Channel Gain}
With the user's relative coordinates established, we can now formulate the PA-to-user channel model.

\begin{prop}[PA-to-User Channel]\label{prop:H_PU}
The $m$-th row and $\ell_n$-th column element of $\bm{H}^{(\text{P} \to \text{U})}$, denoted by $h_{m,\ell_n}^{(\text{P} \to \text{U})}$, represents the channel coefficient from the $\ell_n$-th antenna to the $m$-th user, and is given by
\begin{equation}
h_{m,\ell_n}^{(P \to U)} = \sqrt{\eta} \alpha_{\text{L}}^{|\bm{\psi}_m^{(\text{U})} - \bm{\psi}_{\ell_n}^{(\text{P})}|} \Upsilon\left( {\tilde{x}_{m,\ell_n}^{(\text{U})}}, {\tilde{y}_{m,\ell_n}^{(\text{U})}}, {\tilde{z}_{m,\ell_n}^{(\text{U})}} \right),
\end{equation}
where $\eta = \frac{\lambda^2}{4\pi}$,
$\alpha_{\text{L}}$ denotes the LoS existence coefficient per meter, $\Upsilon$ is the directional radiation pattern of the PA defined in \eqref{e:Upsilon}, and ${\tilde{x}_{m,\ell_n}^{(\text{U})}}$, ${\tilde{y}_{m,\ell_n}^{(\text{U})}}$, ${\tilde{z}_{m,\ell_n}^{(\text{U})}}$ are the coordinates of the $m$-th user relative to the $\ell_n$-th antenna defined in Lemma \ref{lem:psiU_tilde}.
\end{prop}

\begin{proof}
The model constructs the channel coefficient by capturing three physical phenomena: free-space propagation, directional radiation, and stochastic LoS blockage. The term $\sqrt{\eta}$ is fundamental to electromagnetic propagation and represents the effective aperture of an isotropic receiver, accounting for the basic spatial attenuation of the wavefront in free space.
For mmWave frequencies, the presence of a LoS path is probabilistic. Modeling independent blocking events leads to an exponential decay of the LoS probability with distance. Hence, we incorporate the factor $\alpha_{\text{L}}^{|\bm{\psi}_m^{(\text{U})} - \bm{\psi}_{\ell_n}^{(\text{P})}|}$, where $\alpha_{\text{L}}$ is the existence coefficient per unit distance. 
Finally, the $\Upsilon$ term incorporates the directional gain of the PA, which is the core physical feature emphasized in this work.
\end{proof}

\subsubsection{Characteristics and System-Level Implications}
The highly directional nature of the PA radiation is quantified by its beam divergence angles. In the Fraunhofer region, the divergence angles of the radiation field along the $x$- and $z$-directions are expressed as \cite{okamoto2021fundamentals}
\begin{equation}
    \theta_x = \arctan\!\left(\frac{\lambda}{\pi n w_1}\right), \quad
    \theta_z = \arctan\!\left(\frac{\lambda}{\pi n w_2}\right).
\end{equation}
By substituting typical parameters $n = 1.5$, $w_1 = v a \lambda = 1.1 \times 10 \times 10^{-3}$, and 
$w_2 = v b \lambda = 1.1 \times 6 \times 10^{-3}$, 
the corresponding divergence angles are calculated to be 
$\theta_x = 1.11^\circ$ and $\theta_z = 1.84^\circ$. 
This results in an extremely narrow beam. 
For instance, at a distance of $3$~m, the half-power beamwidth therefore covers a region with a diameter of only about $0.1$~m.

\begin{figure}[!tb]
\centering
    \begin{subfigure}[!tb]{1\linewidth}
        \centering
        \includegraphics[width=.8\linewidth]{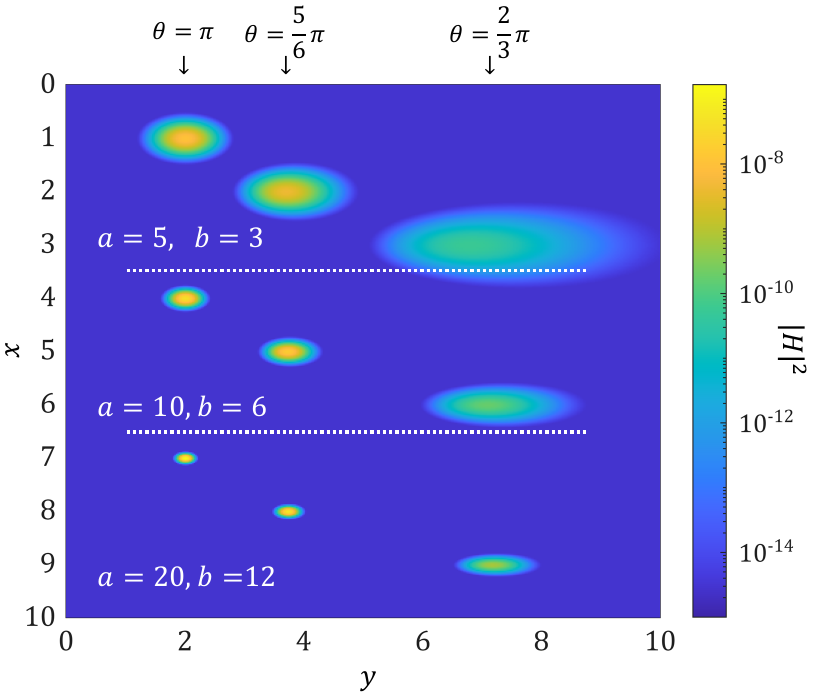}
        \caption{}\label{f:sumrate_M2}
    \end{subfigure}
    \begin{subfigure}[!tb]{1\linewidth}
        \centering
        \includegraphics[width=.8\linewidth]{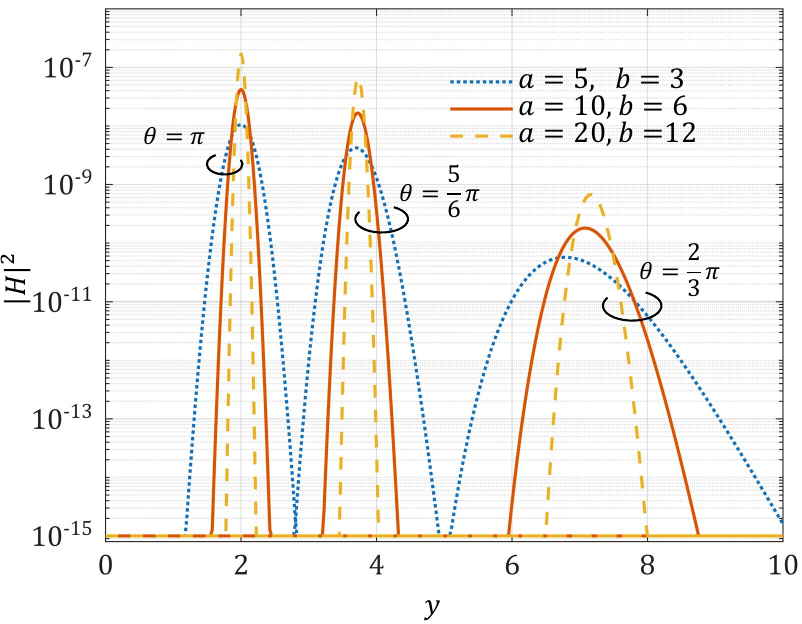}
        \caption{}\label{f:sumrate_M5}
    \end{subfigure}
\caption{Radiation patterns for a set of antenna configurations with varying cross-sectional dimensions ($5\lambda \times 3\lambda$, $10\lambda \times 6\lambda$, and $20\lambda \times 12\lambda$) and elevation angles ($\theta = \pi$, $\tfrac{5}{6}\pi$, and $\tfrac{2}{3}\pi$):
(a) spatial distribution of the power gain $|H|^2$ on the $z = 0$ plane;
(b) variation of the power gain $|H|^2$ along the line aligned with the waveguide axis.}
\label{f:directional}
\end{figure}

To concretely illustrate the extremely limited coverage area of a single PA, we present in Fig.~\ref{f:directional} the radiation patterns for a set of antenna configurations with varying cross-sectional dimensions: $5\lambda \times 3\lambda$, $10\lambda \times 6\lambda$, and $20\lambda \times 12\lambda$. 
The corresponding elevation angles are set to $\pi$, $\tfrac{5}{6}\pi$, and $\tfrac{2}{3}\pi$, respectively. 
All antennas are positioned at $y = 2$ with a fixed azimuth angle of $\tfrac{\pi}{2}$. 
Fig.~\ref{f:directional}(a) presents the spatial distribution of the power gain $|h^{(W \to P)}h^{(P \to U)}|^2$ on the $z = 0$ plane, whereas Fig.~\ref{f:directional}(b) presents the variation of the power gain magnitude along the line aligned with the waveguide axis. 

A key observation is that the effective coverage area of a single PA is highly confined. Even with the smallest simulated cross-section $5\lambda \times 3\lambda$, the effective coverage diameter is less than $1.5$~m. This area further shrinks to approximately $0.5$~m for a larger $20\lambda \times 12\lambda$ cross-section. 
This is a dramatic departure from the omnidirectional model, which assumes uniform radiation in all directions and would significantly overestimate the coverage and underestimate the potential for inter-user interference.
This fundamental property justifies a key system-level assumption in our subsequent analysis: in environments with a user density not exceeding $1~\text{user}/\text{m}^2$, each PA can be assumed to serve only a single user, thereby naturally minimizing inter-user interference.


\section{DiPASS Optimization and Analytical Insights}\label{sec:IV}
The physically consistent model developed in Sections~\ref{sec:II} and~\ref{sec:III} establishes a foundational framework for various PASS-enabled applications. In this section, we focus on the application of PASS in wireless communication and address the fundamental problem of sum-rate maximization. The overall optimization problem is formulated as follows:
\begin{subequations}
\begin{align}
\text{(P1):} 
\max_{\{\!y_{\ell_n}^{(\text{P})}\!\},\{\!\theta_{\ell_n}\!\}, \{\!\varphi_{\ell_n}\!\}, \bm{W}} \!\!& r \! = \! \sum_{m=1}^{M} \! \log_2\! \Big( \!1 \!+\! \frac{|S_m|^2}{|I_m|^2 \!+\! \sigma_m^2} \!\Big), \label{e:P1a}\\
\text{s.t.}
\hspace{2.7 cm} 
 & y_{1_n}^{(\text{P})} > 0, \label{e:P1d}\\
& y_{(\ell+1)_n}^{(\text{P})} - y_{\ell_n}^{(\text{P})} \geq  \frac{\lambda_2}{2}, \label{e:P1e}\\
& y_{L_n}^{(\text{P})} \leq V_y ,\label{e:P1f}\\
& \frac{\pi}{2} \leq \theta_{\ell_n} \leq ~ \pi, \label{e:P1b}\\
&-\pi <  \varphi_{\ell_n} \leq \pi, \label{e:P1c}\\
&\text{tr}\left( \bm{W}^\top \bm{W} \right) = 1 .\label{e:P1w}
\end{align}
\end{subequations}
This formulation jointly optimizes the PA positions $y_{\ell_n}^{(\text{P})}$, orientations $\theta_{\ell_n}$, $\varphi_{\ell_n}$, and the precoding matrix $\bm{W}$ to maximize the sum rate achievable under our physically consistent model. The constraints ensure practical implementation: \eqref{e:P1d} ensures positive placement, \eqref{e:P1e} maintains minimum antenna spacing to avoid inter-PA coupling, \eqref{e:P1f} respects the waveguide length limitation, and \eqref{e:P1b}-\eqref{e:P1c} bound the orientation angles.

Through this optimization framework, we aim to uncover fundamental insights that distinguish our modeling approach from prior work. Specifically, we seek to answer: \textit{what new design principles emerge when accounting for physical effects (i.e., directional radiation, waveguide attenuation, and stochastic blockage) that were oversimplified or ignored in previous analyses?} The solutions to (P1) will reveal how these physical factors collectively influence optimal PASS deployment strategies and potentially challenge conclusions drawn under idealized assumptions.

\subsection{Fundamental Insights from Single-PA Analysis}\label{sec:IVA}
To gain fundamental insights into the competing physical effects, we first analyze the most basic scenario with a single user and a single PA, i.e., $N=L=M=1$. This setup allows for a closed-form analysis that reveals the core trade-offs governing PASS performance. For notational simplicity, we omit all subscripts in this subsection, and (P1) reduces to
\begin{subequations}
\begin{align}
\hspace{-2cm}\text{(P2):}\hspace{3.5cm}& \hspace{-2.5cm}
\max_{\{y^{(\text{P})}\},~\{\theta\},~\{\varphi\}}~
r, \label{e:2a}\\
\hspace{-2cm}\text{s.t.}\hspace{3.7cm}
& \hspace{-2.3cm} \eqref{e:P1d}, \eqref{e:P1f}, \eqref{e:P1b}, \eqref{e:P1c}.
\label{e:P2b}
\end{align}
\end{subequations}

With a single PA and user, neither phase superposition nor interference occurs. Consequently, the sum-rate optimization can be equivalently written as
\begin{align*}
\max r
\Longleftrightarrow & \max \log_2 \left( 1 + \frac{ \left| \sqrt{P} H^{(\text{P} \to \text{U})} H^{(\text{W} \to \text{P})} \right|^2 }{ \sigma^2 } \right) \quad\\
\Longleftrightarrow & \max |H|^2 \triangleq \left| H^{(\text{P} \to \text{U})} H^{(\text{W} \to \text{P})} \right|^2.
\end{align*}

Substituting the channel expressions $H^{(P \to U)}$ and $H^{(W \to P)}$ from Section~\ref{sec:III} yields
\begin{eqnarray}\label{e:H_2}
|H|^2 &=& \left| H^{(\text{P} \to \text{U})} H^{(\text{W} \to \text{P})} \right|^2 \nonumber \\
&=& e^{ -\alpha_{\text{W}} y^{(\text{P})} } \alpha_{\text{L}}^{ 2\sqrt{ \left( {\tilde{x}^{(\text{U})}} \right)^2 + \left( {\tilde{y}^{(\text{U})}} \right)^2 + \left( {\tilde{z}^{(\text{U})}} \right)^2 } } \frac{ n^2 v^2 a b \lambda^2 }{ 2 \left( {\tilde{y}^{(\text{U})}} \right)^2 } \nonumber \\
\hspace{-1cm}&& \cdot \exp\left\{ -2 (\pi n v)^2 \frac{ \left( a {\tilde{x}^{(\text{U})}} \right)^2 + \left( b {\tilde{z}^{(\text{U})}} \right)^2 }{ \left( {\tilde{y}^{(\text{U})}} \right)^2 } \right\}.
\end{eqnarray}

An illustration of $|H|^2$ versus the PA position $y^{\text{P}}$ is given in Fig.~\ref{fig:Hp_y}, wherein users are positioned at $y^{(\text{U})} = 2$, $5$, and $8$, respectively. For $y^{(\text{U})} = 2$, $|H|^2$ decreases monotonically as the PA moves leftward from $y = 0$, and the rate of decrease transitions from gradual to exponential. In contrast, when $y^{(\text{U})} = 5$ or $8$, $|H|^2$ initially increases before subsequently declining. These results reveal the existence of an optimal PA position that maximizes $|H|^2$. Nevertheless, as indicated by the expression in \eqref{e:H_2}, deriving this optimal position analytically remains a nontrivial task. 
To circumvent this challenge, we employ the logarithmic function to construct a positively correlated term $\ln |H|^2$ shown as 
\begin{eqnarray}\label{e:ln_H_2}
&&\hspace{-1cm}\ln {|H|^2} 
= - \alpha_{\text{W}} y^{(\text{P})} + 2\ln(\alpha_{\text{L}})\\
&&  \cdot \sqrt{ \left( {\tilde{x}^{(\text{U})}} \right)^2 + \left( {\tilde{y}^{(\text{U})}} \right)^2 + \left( {\tilde{z}^{(\text{U})}} \right)^2 } + \ln\left( n^2 v^2 a b \lambda^2 \right) \notag \\
&&  - 2 \ln\left( \sqrt{2} {\tilde{y}^{(\text{U})}} \right)
- 2 (\pi n v)^2 \frac{ \left( a {\tilde{x}^{(\text{U})}} \right)^2 + \left( b {\tilde{z}^{(\text{U})}} \right)^2 }{ \left( {\tilde{y}^{(\text{U})}} \right)^2 }, \notag
\end{eqnarray}
which transforms the multiplicative relationship into an additive form and facilitates analysis.

To determine the optimal $y^{(\text{P})}$, $\theta$, and $\varphi$, we observe that the local coordinates ${\tilde{x}}$, ${\tilde{y}}$, ${\tilde{z}}$ are analytical functions of these optimization variables. We leverage the chain rule via the Jacobian matrix $\bm{J}$:
\begin{eqnarray}
\begin{bmatrix} \frac{\partial \ln |H|^2}{\partial y^{(\text{P})}} \\ \frac{\partial \ln |H|^2}{\partial \theta} \\ \frac{\partial \ln |H|^2}{\partial \varphi} \end{bmatrix} \hspace{-0.2cm}&=&\hspace{-0.2cm} \begin{bmatrix}
\frac{\partial {\tilde{x}}}{\partial y^{(\text{P})}} & \frac{\partial {\tilde{y}}}{\partial y^{(\text{P})}} & \frac{\partial {\tilde{z}}}{\partial y^{(\text{P})}} \\
\frac{\partial {\tilde{x}}}{\partial \theta} & \frac{\partial {\tilde{y}}}{\partial \theta} & \frac{\partial {\tilde{z}}}{\partial \theta} \\
\frac{\partial {\tilde{x}}}{\partial \varphi} & \frac{\partial {\tilde{y}}}{\partial \varphi} & \frac{\partial {\tilde{z}}}{\partial \varphi}
\end{bmatrix}
\cdot \begin{bmatrix} \frac{\partial \ln |H|^2}{\partial {\tilde{x}}} \\ \frac{\partial \ln |H|^2}{\partial {\tilde{y}}} \\ \frac{\partial \ln |H|^2}{\partial {\tilde{z}}} \end{bmatrix}  \notag\\
\hspace{-0.2cm}&\triangleq&\hspace{-0.2cm} \bm{J} \cdot \begin{bmatrix} \frac{\partial \ln |H|^2}{\partial {\tilde{x}}} , \frac{\partial \ln |H|^2}{\partial {\tilde{y}}} , \frac{\partial \ln |H|^2}{\partial {\tilde{z}}} \end{bmatrix}^\top.
\end{eqnarray}

From Lemma \ref{lem:psiU_tilde}, we can obtain the explicit expression of $\bm{J}$, which is  full-rank. Therefore,
\begin{equation*}
\begin{bmatrix} \frac{\partial \ln |H|^2}{\partial y^{(\text{P})}} \\ \frac{\partial \ln |H|^2}{\partial \theta} \\ \frac{\partial \ln |H|^2}{\partial \varphi} \end{bmatrix} = \begin{bmatrix} 0 \\ 0 \\ 0 \end{bmatrix} \Leftrightarrow \begin{bmatrix} \frac{\partial \ln |H|^2}{\partial {\tilde{x}}} \\ \frac{\partial \ln |H|^2}{\partial {\tilde{y}}} \\ \frac{\partial \ln |H|^2}{\partial {\tilde{z}}} \end{bmatrix} = \begin{bmatrix} 0 \\ 0 \\ 0 \end{bmatrix}.
\end{equation*}

This allows us to solve for the optimal orientation and position separately.

\begin{figure}[!tb]
    \centering
    \includegraphics[width=0.8\linewidth]{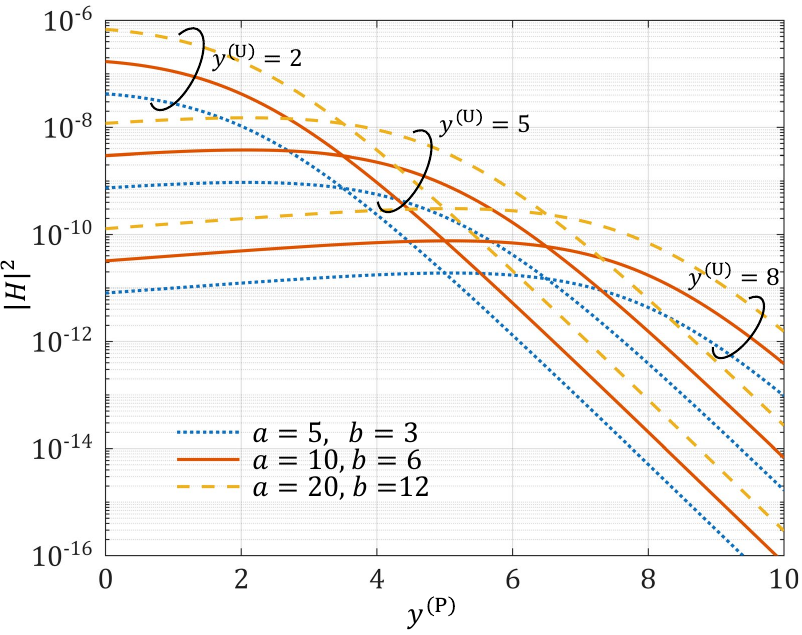}
    \caption{ Variation of the power gain $|H|^2$ with respect to the PA position $y^{(\text{P})}$ under different waveguide cross-sectional dimensions and user positions.}
    \label{fig:Hp_y}
\end{figure}


\begin{lem}[Optimal PA Orientation]\label{thm:single_angle}
The optimal orientation angles that maximize the sum rate in (P2) are given by
\begin{equation} \label{e:theta}
    \theta^* = \arctan\left( \frac{ \sqrt{ \left( y^{(\text{U})} - y^{(\text{P})} \right)^2 + \left( x^{(\text{U})} - x^{(\text{P})} \right)^2 } }{ z^{(\text{U})} - z^{(\text{P})} } \right),
\end{equation}
\begin{equation} \label{e:phi}
\varphi^* = \arctan\left( \frac{ y^{(\text{U})} - y^{(\text{P})} }{ x^{(\text{U})} - x^{(\text{P})} } \right),
\end{equation}
in which case the local coordinates are
\begin{equation*}
    \tilde{x}^{(\text{U})} = \tilde{z}^{(\text{U})} = 0,~~
    \tilde{y}^{(\text{U})} =\left| \bm{\psi}^{(\text{U})} - \bm{\psi}^{(\text{P})} \right|.
\end{equation*}

\end{lem}

\begin{proof}
Setting $\frac{\partial}{\partial {\tilde{x}^{(\text{U})}}} \ln |H|^2 = 0$ yields
\begin{eqnarray*}
\left[ \frac{ 2\ln(\alpha_{\text{L}})}{\sqrt{ \left( {\tilde{x}^{(\text{U})}} \right)^{2} + \left( {\tilde{y}^{(\text{U})}} \right)^{2} + \left( {\tilde{z}^{(\text{U})}} \right)^{2} }} 
 - \frac{ 4 (\pi n v)^2 }{ \left( {\tilde{y}^{(\text{U})}} \right)^2 } a^2 \right] {\tilde{x}^{(\text{U})}}=0.
\end{eqnarray*}
Since $\ln(\alpha_{\text{L}}) < 0$, the first multiplicative term in the above expression is always negative. Moreover, the left-hand side increases when ${\tilde{x}^{(\text{U})}} < 0$ and decreases when ${\tilde{x}^{(\text{U})}} > 0$. Therefore, the maximum is attained at $\left( {\tilde{x}^{(\text{U})}} \right)^* = 0$ when ${\tilde{y}^{(\text{U})}} \neq 0$.

Substituting the specific expression of $\left( {\tilde{x}^{(\text{U})}} \right)^*$ in Lemma \ref{lem:psiU_tilde} into $\left( {\tilde{x}^{(\text{U})}} \right)^* = 0$, we have 
\begin{equation*}
\left( {\tilde{x}^{(\text{U})}} \right)^* = \sin \varphi \left( x^{(\text{U})} - x^{(\text{P})} \right) - \cos \varphi \left( y^{(\text{U})} - y^{(\text{P})} \right) = 0.
\end{equation*}
Thus, we can solve the optimal $\varphi^*$.

Similarly, for the second equation $\frac{\partial}{\partial {\tilde{z}^{(\text{U})}}} \ln |H|^2 = 0$,   we get $\left( {\tilde{z}^{(\text{U})}} \right)^* = 0$, when ${\tilde{y}^{(\text{U})}} \neq 0$.

Combining the explicit expression of $\left( {\tilde{z}^{(\text{U})}} \right)^*$ in Lemma \ref{lem:psiU_tilde}, the condition $\left( {\tilde{z}^{(\text{U})}} \right)^* = 0$, and \eqref{e:phi}, the optimal value of $\theta$ can be derived as 
\begin{equation*}
    \theta^* = \arctan\left( \frac{ \sqrt{ \left( y^{(\text{U})} - y^{(\text{P})} \right)^2 + \left( x^{(\text{U})} - x^{(\text{P})} \right)^2 } }{ z^{(\text{U})} - z^{(\text{P})} } \right).
\end{equation*}


Given $\left( {\tilde{x}^{(\text{U})}} \right)^* = \left( {\tilde{z}^{(\text{U})}} \right)^* = 0$, the optimal $\left( {\tilde{y}^{(\text{U})}} \right)^*$ is given by
\begin{eqnarray}
\left( {\tilde{y}^{(\text{U})}} \right)^* =\left| \bm{\psi}^{(\text{U})} - \bm{\psi}^{(\text{P})} \right|,
\end{eqnarray}
thereby completing the proof.
\end{proof}

Lemma \ref{thm:single_angle} indicates an intuitive result: the optimal orientation simply directs the PA's beam perfectly towards the user to maximize the directional gain. With the optimal orientation established, we now determine the optimal PA position $\left( y^{(\text{P})} \right)^*$ along the waveguide.

We first analyze the derivative of the channel gain with respect to $y^{(\text{P})}$:
\begin{eqnarray}\label{e:partial_ln_H_2}
\hspace{-1cm}&&\frac{\partial \ln |H|^2}{\partial y^{(\text{P})}}
\overset{(a)}{=} -\alpha_{\text{W}} + \ln(\alpha_{\text{L}}) \cdot 2 \frac{\partial {\tilde{y}^{(\text{U})}}}{\partial y^{(\text{P})}} - 2 \frac{1}{ {\tilde{y}^{(\text{U})}} } \frac{\partial {\tilde{y}^{(\text{U})}}}{\partial y^{(\text{P})}} \nonumber \\
\hspace{-1cm}&&\overset{(b)}{=} -\alpha_{\text{W}} - 2 \ln(\alpha_{\text{L}})  \frac{ y^{(\text{U})} - y^{(\text{P})} }{ \left| \bm{\psi}^{(\text{U})} - \bm{\psi}^{(\text{P})} \right| } + 2   \frac{ y^{(\text{U})} - y^{(\text{P})} }{ \left| \bm{\psi}^{(\text{U})} - \bm{\psi}^{(\text{P})} \right|^2 },
\end{eqnarray}
where (a) follows by substituting $\left( {\tilde{x}^{(\text{U})}} \right)^* = \left( {\tilde{z}^{(\text{U})}} \right)^* = 0$,
(b) follows by substituting $ \tilde{y}^{(\text{U})} $ and $\frac{\partial {\tilde{y}^{(\text{U})}}}{\partial y^{(\text{P})}}$.

\begin{rem}\label{rem:yP_leq_yU}
We can reasonably assume that $\left( y^{(\text{P})} \right)^* \in (0,y^{(\text{U})})$, because if $\left( y^{(\text{P})} \right)^* > y^{(\text{U})} $, we can always find a symmetric position $\left( y^{(\text{P})} \right)^{\prime}$ for the PA with respect to $y^{(\text{U})}$, such that $\left( y^{(\text{P})} \right)^* - y^{(\text{U})} = y^{(\text{U})} - \left( y^{(\text{P})} \right)^{\prime}$. The required wireless space distance is the same, but since $\left( y^{(\text{P})} \right)^{\prime} < \left( y^{(\text{P})} \right)^*$, it only needs to travel a shorter distance within the waveguide.
\end{rem}

The three terms in \eqref{e:partial_ln_H_2} represent competing effects:
\begin{itemize}
    \item The impact of waveguide attenuation: as $y^{(\text{P})}$ increases, the waveguide attenuation steadily increases.
    \item The impact of LOS path existence probability: as $y^{(\text{P})}$ increases, the LOS existence probability increases, causing $|H|^2$ to increase at a decelerating rate.
    \item The impact of transmission in wireless space: as $y^{(\text{P})}$ increases, the distance required for transmission in free space becomes shorter, and the corresponding attenuation decreases.
\end{itemize}

Thus, there exists a clear trade-off between waveguide loss and wireless propagation reliability. This leads to a necessary condition for the optimal position to exists.


\begin{thm}[Existence of the Optimal PA Position]\label{thm:condition}
For a waveguide with attenuation coefficient $\alpha_{\text{W}}$ in an environment with LOS existence coefficient $\alpha_{\text{L}}$, the optimal PA position exists only when 
\begin{equation}\label{eq:THMcondition}
\alpha_{\text{L}} < \exp\left(-\frac{\alpha_{\text{W}}}{2}\right).
\end{equation}
\end{thm}

\begin{proof}
See Appendix \ref{sec:AppA}.
\end{proof}

Theorem~\ref{thm:condition} establishes the quantitative trade-off between waveguide attenuation and wireless propagation reliability. If this condition is not met, any additional distance traveled within the waveguide results only in further energy loss; otherwise, extending the in-waveguide propagation distance can yield a performance gain. The amount of extension that maximizes this gain is characterized in Theorem~\ref{thm:optimal_y}.

\begin{thm}[Optimal PA Position] \label{thm:optimal_y}
If the condition in Theorem \eqref{thm:condition} is satisfied, the optimal PA position is given by
\begin{equation}\label{e:optimal_y}
y^* = \max {\left\{0, y^U - \gamma^* \right\}},
\end{equation}
where $$\gamma^* = \frac{\sqrt{A}}{\tan\left(\dfrac{-1 \! - \! \sqrt{A \! + \! \ln \alpha_{\text{L}} A^2 (\alpha_{\text{W}} \! +\!  2 \ln \alpha_{\text{L}})}}{\ln \alpha_{\text{L}} \cdot A}\right)}$$ is the optimal horizontal offset of the PA relative to the user, and $A = (x^{(\text{U})} - x^{(\text{P})})^2 + (z^{(\text{U})} - z^{(\text{P})})^2$. 
An approximation of $\gamma^*$ is
\begin{equation}\label{e:optimal_y_approx}
\gamma^* = \sqrt{\frac{A\alpha_{\text{W}}^2}{(2 \ln \alpha_{\text{L}})^2 - \alpha_{\text{W}}^2}}.
\end{equation}

\end{thm}

\begin{proof}
See Appendix \ref{sec:AppB}.
\end{proof}

\begin{rem}[Translation-Invariant Optimal Offset]
    The optimal horizontal offset of the PA relative to the user, $\gamma^*$, is a constant determined solely by the system parameters $\alpha_{\text{W}}$, $\alpha_{\text{L}}$, and the squared lateral distance $A$ between user and waveguide. This reveals a translation-invariant structure: the optimal PA position $y^*$ always maintains a fixed offset from the user location $y^{(\text{U})}$, implying that the optimal offset does not incur accumulated attenuation as the user moves along the waveguide direction.
\end{rem}

\begin{rem}[Fundamental Trade-off Governing PA Placement]
    Theorem \ref{thm:optimal_y} quantifies the core trade-off: the optimal PA position shifts away from the user (i.e., $y^*$ decreases) as waveguide attenuation $\alpha_{\text{W}}$ decreases or as the LoS environment improves (larger $\alpha_{\text{L}}$). This is because lower waveguide loss or a less reliable LoS link makes it beneficial to ``spend'' more distance in the waveguide to get closer to the user in free space.
\end{rem}

\begin{rem}[Practical Implications for Waveguide Deployment]
    For a square coverage region of length $D_y$, even if full-area coverage is required, only a waveguide of length $D_y - \gamma^*$ is needed. Alternatively, if the feeding point of the waveguide is placed at its midpoint, a length of $D_y - 2\gamma^*$ would suffice. Moreover, by employing segmented waveguides\cite{ouyang2025uplink}, the deployment efficiency can be further improved.
\end{rem}

\begin{rem}[Reconciliation with Prior Idealized Models]
    According to \eqref{e:optimal_y_approx}, when the user is collinear with the waveguide or when the waveguide attenuation coefficient is zero, the optimal position coincides with $y^{(\text{U})}$, which is consistent with the results reported in prior works \cite{liu2025pinching} that assume no attenuation within the waveguide.
\end{rem}

\subsection{Multi-PA Optimization Framework}\label{sec:IVB}

\begin{algorithm}[!tb]
\caption{Joint Optimization for DiPASS.}\label{alg:joint_optimization}
\begin{algorithmic}[1]
\vspace{0.1cm}
\REQUIRE $\{\bm{\psi}_m^{(\text{U})}\}$, $\{\bm{\psi}_n^{(\text{W})}\}$, $D_y$, $\alpha_{\text{W}}$, $\alpha_{\text{L}}$
\ENSURE $\bm{\mathcal{M}}^*$, $\bm{W}^*$, $\{y_{\ell_n}^{(\text{P})*}\}$

\STATE Initialize $\bm{H}^{(\text{O})}$ by \eqref{e:HP}
\STATE \textbf{if} $M \geq NL$ 
\STATE \hspace{0.3cm} Solve $\bm{\mathcal{M}}^*$ by Hungarian algorithm with virtual PAs
\STATE \textbf{else}
\STATE \hspace{0.3cm} Initialize $\bm{\mathcal{M}}^{(M)}$ with virtual users
\STATE \hspace{0.3cm} \textbf{for} $j = 1, \dots, NL-M $ 
\STATE \hspace{0.6cm} Greedy allocation for remaining PAs by \eqref{e:marginal_gain}
\STATE \hspace{0.6cm} Update $\bm{\mathcal{M}}^{M+j}$
\STATE Verify constraints \eqref{e:P1d}, \eqref{e:P1e}, \eqref{e:P1f} 
\STATE Phase Fine-Tuning
\STATE Solve $\bm{W}^*$ by beamforming algorithm 
\STATE \textbf{return} $\bm{\mathcal{M}}^*, \bm{W}^*, \{y_{\ell_n}^{(\text{P})*}\}$
\end{algorithmic}
\end{algorithm}

Building upon the fundamental insights from the single-PA analysis, we now address the general sum-rate maximization problem (P1) for systems with multiple waveguides and PAs. Our prior analysis established two key insights:
\begin{itemize}[leftmargin=0.5cm]
    \item The highly directional nature of PA radiation confines its effective coverage, making it practical to assume that each PA primarily serves a single user to minimize inter-user interference. 
    \item For any given PA-user pair, the optimal orientation and position that maximize the channel gain have been derived in closed-form in Lemma~\ref{thm:single_angle} and Theorem~\ref{thm:optimal_y}, respectively.
\end{itemize}

These premises allow us to recast the complex joint optimization in (P1) into a more tractable PA-user assignment problem. The core idea is to pre-compute the maximum possible channel gain each PA can offer to every user when configured at its optimal position and orientation. We encapsulate these gains into an optimal channel gain matrix $\bm{H}^{(\text{O})} \in \mathbb{R}^{M \times NL}$, where each element is given by
\begin{eqnarray}\label{e:HP}
\hspace{-1cm}&& h^{(\text{O})}_{m,\ell_n} 
= \sqrt{\frac{1}{L}} e^{ -\frac{\alpha_{\text{W}}}{2} \left( y_{\ell_n}^{(\text{P})} \right)^* } \alpha_{\text{L}}^{ \left( {\tilde{y}_{m,\ell_n}^{(\text{U})}} \right)^* } \frac{\lambda}{2} \frac{ n v \sqrt{2 a b} }{ {\tilde{y}_{m,\ell_n}^{(\text{U})}} } \\
\hspace{-1cm}&& \cdot \exp\left\{ -j \frac{2\pi}{\lambda_g}  \left( y_{\ell_n}^{(\text{P})} \right)^* - j k n \left( {\tilde{y}_{m,\ell_n}^{(\text{U})}} \right)^* + j \frac{\Theta_1 + \Theta_2}{2} \right\}. \nonumber
\end{eqnarray}
This composite channel gain is synthesized from the two channel gain given in Propositions \ref{lem:H_WP} and \ref{prop:H_PU}, and the optimal distance $\left( {\tilde{y}_{m,\ell_n}^{(\text{U})}} \right)^* $ given in Theorem \ref{thm:optimal_y}.

With $\bm{H}^{(\text{O})}$ defined, we formalize the assignment problem by introducing a binary mask matrix $\bm{\mathcal{M}}$ of the same dimensions. An entry $\mathcal{M}_{m,\ell_n}=1$ signifies the assignment of PA $\ell_n$ to user $m$. The constraints ensure a valid matching: each PA is assigned to at most one user (Eq. \eqref{e:P3_assign1}), and when resources permit ($M \leq N L$), each user is guaranteed service (Eq. \eqref{e:P3_assign2}). This leads to the reformulated problem as follows.
\begin{subequations}
\begin{align}
\hspace{-1cm}\text{(P3):} \quad  \max_{\bm{\mathcal{M}},~\bm{W}} F(\bm{\mathcal{M}}) &= \sum_{m=1}^{M} \log_2 \!\left(\! 1 \!+\!  \frac{ |q_{m,1}|^2 }{ |q_{m,2}|^2 \!+\! \sigma_m^2 } \!\right)\!\!, \label{e:P3obj}\\
\text{s.t.} \hspace{2.7cm} & \sum_{m=1}^{M} \mathcal{M}_{m,\ell_n} = 1, \label{e:P3_assign1}\\
& \sum_{\ell_n=1}^{NL} \mathcal{M}_{m,\ell_n} \geq 1, \label{e:P3_assign2}\\
& \bm{\mathcal{M}} \in \{0,1\}^{M \times NL}, \label{e:P3_binary}\\
& \eqref{e:P1w}.
\end{align}
\end{subequations}

Here, $|q_{m,1}|^2$ and $|q_{m,2}|^2$ represent the signal and interference power for user $m$, respectively:
\begin{align}
|q_{m,1}|^2 &= P \left| \sum_{\ell_n=1}^{NL} \mathcal{M}_{m,\ell_n} \left| [\bm{H}^{(\text{O})}]_{m,\ell_n} \right| [\bm{W}]_{n,m} \right|^2, \\
|q_{m,2}|^2 &=\! P \!\!\!\! \!\sum_{i=1, i \neq m}^{M} \left| \sum_{\ell_n=1}^{NL} \mathcal{M}_{m,\ell_n} \left| [\bm{H}^{(\text{O})}]_{m,\ell_n} \right| [\bm{W}]_{n,i} \right|^2.
\end{align}

This formulation leverages the fact that phase can be adjusted over a $2\pi$ range at the wavelength scale with negligible impact on signal intensity. Furthermore, the directional radiation of PAs naturally mitigates interference, allowing individual PA position adjustments for phase alignment. This will be addressed in the final step of our optimization.

\subsubsection{PA Assignment}
The task of solving (P3) is to find the optimal assignment matrix $\bm{\mathcal{M}}$. 
As summarized in Algorithm \ref{alg:joint_optimization}, we initialize $\bm{\mathcal{M}}$ as an $M \times NL$ all-zero matrix and distinguish two cases: $M \geq NL$ and $M < NL$.

In the first case, with users outnumbering PAs, the goal is to select the best $NL$ user-PA pairings. We introduce $M-NL$ virtual PAs with zero channel gain and apply the Hungarian algorithm \cite{papadimitriou1998combinatorial} to the augmented problem, maximizing the sum rate while ensuring each real PA serves one user.

In the first case, users outnumbering PAs. Since each antenna can effectively cover only one user in space, at most $NL$ users can be served simultaneously. To enforce this constraint, we add $M - NL$ virtual PAs with zero channel gain and apply the Hungarian algorithm \cite{papadimitriou1998combinatorial} to the augmented assignment problem for sum rate maximization. The resulting assignment matrix $\bm{\mathcal{M}}$ contains exactly $NL$ columns with a single non-zero entry, thereby achieving the highest possible sum rate under the given constraints.

In the second case, where $M < NL$, we ensure that each user is assigned at least one antenna by first adding $(NL - M)$ virtual users with zero channel gain, so that the total number of users matches the number of PAs. This allows us to apply the Hungarian algorithm again to obtain the optimal matching. The resulting assignment matrix, denoted by $\bm{\mathcal{M}}^{(M)}$, contains exactly $M$ columns with a single $1$, indicating that each real user is assigned one PA in this initial stage.

The remaining antennas can then be allocated iteratively using a greedy algorithm. For each iteration $j = 1 \text{ to } (NL - M)$, we start with the current assignment matrix $\bm{\mathcal{M}}^{(M+j-1)}$, which already contains $(M + j - 1)$ assigned PAs. For every unassigned PA $\ell_n$ and each user $m$, we compute the marginal rate gain obtained by assigning the $(M+j)$-th PA to that user, which is expressed as follows:
\begin{equation}\label{e:marginal_gain}
    \Delta r_m^{(j)} = \log_2\!\left( \frac{1 + \xi_m^{(j)}}{1 + \xi_m^{(j-1)}} \right),
\end{equation}
where $\xi_m^{(j)}$ denotes the SINR after incorporating the additional PA. The PA-user pair that yields the maximum $\Delta r_m^{(j)}$ is selected, and the assignment matrix is updated as $\bm{\mathcal{M}}^{(M+j)}$ accordingly. After completing all $(NL - M)$ iterations, the final assignment matrix is obtained, in which all $NL$ PAs have been assigned.

The SINR after the $j$-th assignment is computed as
\begin{equation*}
\xi_m^{(j)} = \frac{ P \left| \sum_{\ell_n=1}^{NL} \mathcal{M}_{m,\ell_n}^{(j)} \left| h^{(\text{O})}_{m,\ell_n} \right| [\bm{W}]_{n,m} \right|^2 }{ P \sum_{i=1, i \neq m}^{M} \left| \sum_{\ell_n=1}^{NL} \mathcal{M}_{m,\ell_n}^{(j)} \left| h^{(\text{O})}_{m,\ell_n} \right| [\bm{W}]_{n,i} \right|^2 + \sigma_m^2 },
\end{equation*}
where the summation is taken over all PAs $\ell_n = 1$ to $NL$, but only those with $\mathcal{M}_{m,\ell_n}^{(j)} = 1$ contribute to the signal and interference terms.

\subsubsection{Post-Assignment Fine-tuning and Beamforming} 
The assignment algorithm determines the optimal serving relationship between PAs and users. Based on this assignment, the nominal optimal position for each PA, as derived from Theorem \ref{thm:optimal_y}, is obtained. However, these nominal positions may not form a physically realizable system configuration. Therefore, we introduce two fine-tuning stages to ensure practicality and maximize performance.

First, the independently computed optimal PA positions on the same waveguide may violate the minimum spacing constraint \eqref{e:P1e} or exceed the waveguide length limit \eqref{e:P1d}. To address this, we process each waveguide sequentially: for any PA whose nominal position violates these physical constraints, its position is projected onto the boundary of the feasible region. This ensures that the final configuration is deployable without compromising the assignment solution.

Second, when multiple PAs from different waveguides are assigned to the same user, their signals can be coherently combined at the receiver. We exploit this property by fine-tuning the PA positions at the wavelength scale to further enhance coherent combining gain.

According to \eqref{e:HP}, the phase of the channel is
\begin{equation}
    \phi_{m,\ell_n}^{(\text{O})} = -j \frac{2\pi n}{\lambda} \left( y_{\ell_n}^{(\text{P})} \right)^* - j \frac{2\pi}{\lambda} \left( {\tilde{y}_{m,\ell_n}^{(\text{U})}} \right)^* + j \frac{\Theta_1 + \Theta_2}{2}.
\end{equation}
Therefore, A small displacement $\Delta y_{\ell_n}$ of the PA induces a phase change of
\begin{equation}
\Delta \phi_{m,\ell_n}^{(\text{O})}  =  -j \left( \frac{2\pi n}{\lambda} - \frac{2\pi}{\lambda}  \sin \theta_{\ell_n} \sin \varphi_{\ell_n} \right) \Delta y_{\ell_n} .
\end{equation}

Select the path with the strongest signal to user $m$ as the reference, and align the phases of the remaining paths with it by adjusting each PA within a displacement of one wavelength, $\Delta y_{\ell_n}$. If the adjusted PA position violates the distance constraint, it is projected back into the feasible region.

Finally, with the physical configuration (PA assignment, positions, and orientations) optimized, the digital precoding matrix $\bm{W}$ can be computed using established algorithms like WMMSE \cite{shi2011iteratively} or ZF \cite{wiesel2008zero} to mitigate any residual multi-user interference at the base station. The complete optimization procedure is outlined in Algorithm \ref{alg:joint_optimization}.

\section{Numerical Results}\label{sec:V}
This section provides a comprehensive performance evaluation of the proposed directional-PASS framework through extensive numerical simulations. We validate the analytical derivations and demonstrate the system's capabilities in terms of coverage characteristics and sum-rate optimization. Unless otherwise specified, all simulations are conducted at a carrier frequency of $100$~GHz, with a waveguide attenuation coefficient of $\alpha_{\text{W}} = 1.3$ dB/m, a LOS existence coefficient of $\alpha_{\text{L}} = 0.5$, and users uniformly distributed on the floor of a $10\,\text{m} \times 10\,\text{m} \times 3\,\text{m}$ service volume. The total transmit power is set to $40$~W.



\begin{figure}[!tb]
    \centering
    \includegraphics[width=0.8\linewidth]{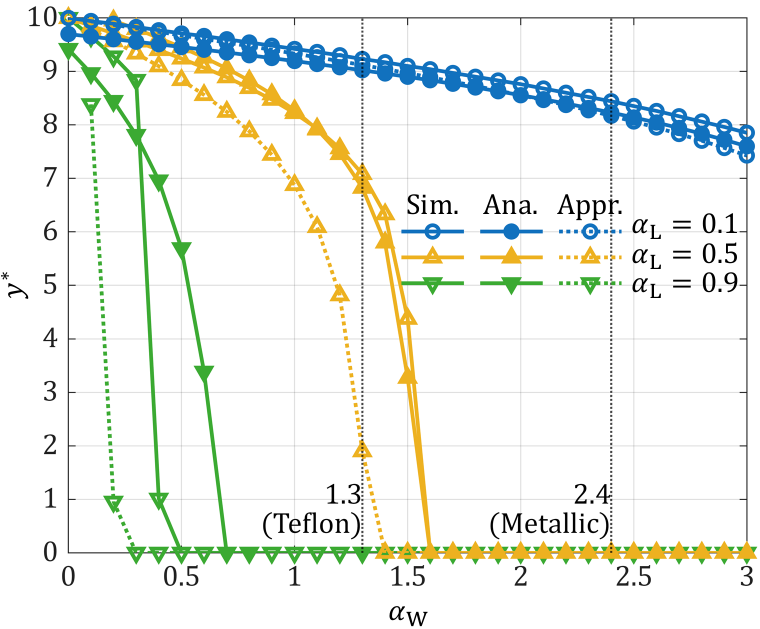}
    \caption{ Optimal $y^*$ under different combinations of $\alpha_{\text{W}}$ and $\alpha_{\text{L}}$, with simulation results, analytical solutions in \eqref{e:optimal_y}, and simplified approximate expressions in \eqref{e:optimal_y_approx}.}
    \label{fig:optimal_y}
\end{figure}

\subsection{Single-Antenna Performance Analysis}
We begin by examining the fundamental trade-offs in a single-antenna, single-user scenario to validate our theoretical analysis and reveal key physical insights.

Fig.~\ref{fig:optimal_y} compares the optimal PA position $y^*$ obtained through numerical optimization with our analytical solution in \eqref{e:optimal_y} and its approximation in \eqref{e:optimal_y_approx}. In the comparison, the single user is located at $y = 10$. 
It can be observed that \eqref{e:optimal_y} accurately captures the variation of the optimal position across different parameter settings. 
While the approximate solution \eqref{e:optimal_y_approx} exhibits some expected deviation, it correctly captures the essential trends of how $y^*$ varies with system parameters.

The results provides an intuitive illustration of how $y^*$ varies with the three key parameters.
The material employed in Docomo's demonstration is Polytetrafluoroethylene (PTFE) \cite{suzuki2022pinching}, commonly known as Teflon, whose typical attenuation is approximately 1.3~dB/m \cite{yeh2008essence}. For reference, a conventional metallic waveguide exhibits a higher attenuation of about 2.4~dB/m. Both attenuation values are indicated in the figure for comparative clarity.
When the waveguide attenuation is small, the additional propagation distance inside the waveguide does not incur a significant cost. Similarly, when the LOS existence probability is low, it is preferable to minimize the propagation distance in free space. In both cases, the optimal antenna position tends to be closer to the user. In the extreme case where the waveguide attenuation is zero, the optimal position simplifies to the location closest to the user.

The variation of $y^*$ with different deployment heights is further illustrated in Fig.~\ref{fig:y_aWG}. As the height difference increases, the benefit in wireless transmission distance from the same antenna displacement decreases, reflecting a weaker near-field effect. Consequently, incurring additional propagation loss inside the waveguide becomes less cost-effective, leading to a more conservative choice of antenna position. Moreover, $y^*$ becomes increasingly sensitive to variations in the waveguide attenuation.
As shown, when the waveguide attenuation is $\alpha_{\text{W}} = 1.3$~dB/m and the deployment height is $z = 3$, the optimal antenna position $y^*$ remains within $7$m when the LOS visibility is not less than $0.5$. This implies that a $7$m-long waveguide is sufficient to provide full coverage for a $10$m region. Furthermore, when the deployment height increases to $z = 10$ and the LOS existence probability exceeds $0.3$, extending the waveguide beyond $3$m does not yield any additional performance gain.

\begin{figure}[!tb]
    \centering
    \includegraphics[width=0.8\linewidth]{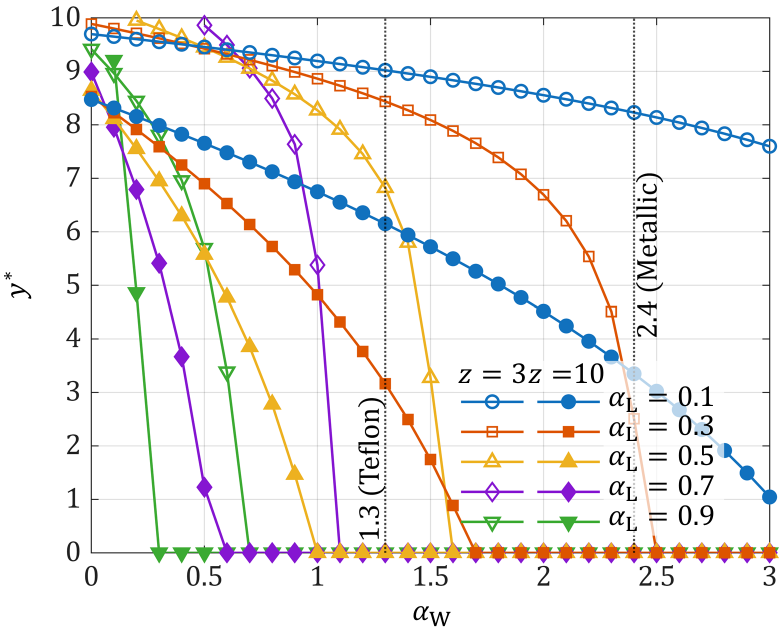}
    \caption{ Variation of $y^*$ under different deployment heights.}
    \label{fig:y_aWG}
\end{figure}

Fig.~\ref{f:map_y_star} illustrates the spatial distributions of $y^*$ and channel gain $|H|^2$ under the single-waveguide and single-antenna configuration. It can be observed that the directional radiation model enables significantly wide coverage within the waveguide length. As shown in Theorem~\ref{thm:optimal_y}, for a given system setup, the optimal relative position remains constant. Along the direction parallel to the waveguide, $y^*$ varies linearly with the user's coordinate, allowing approximately $30\%$ of the waveguide length to cover most of the service area. Along the parallel line, the channel power gain $|H|^2$ decreases approximately exponentially, mainly due to the intrinsic attenuation inside the waveguide. The maximum attenuation is around 50~dB, which is comparable to the typical loss of indoor 2.4~GHz WiFi. In contrast, the attenuation of the 100~GHz PASS system proposed in this work remains lower than that of conventional WiFi in most regions. Moreover, by guiding the signal closer to the user through the waveguide, the proposed system enhances the likelihood of maintaining a line-of-sight (LOS) connection, while the use of high-frequency signals inherently mitigates multipath interference.

\begin{figure}[!tb]
\centering
    \begin{subfigure}[!tb]{1\linewidth}
        \centering
        \includegraphics[width=.8\linewidth]{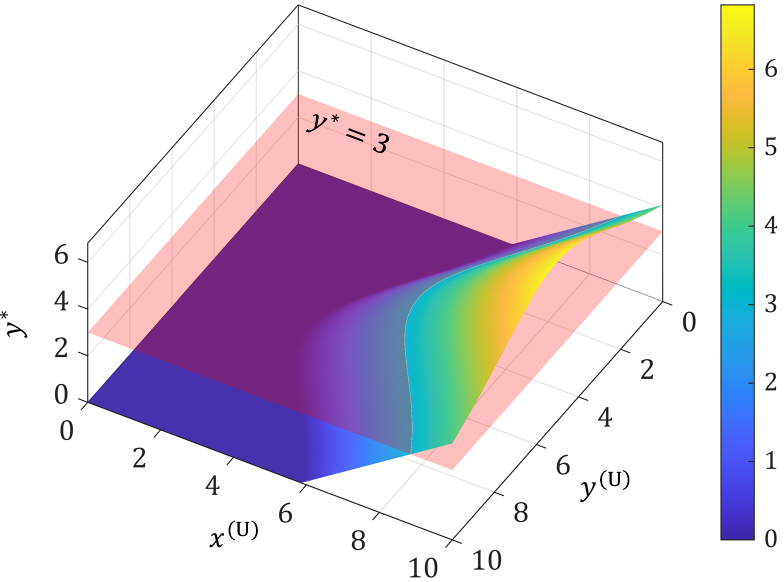}
        \caption{}
    \end{subfigure}
    \hfill
    \begin{subfigure}[!tb]{1\linewidth}
        \centering
        \includegraphics[width=.8\linewidth]{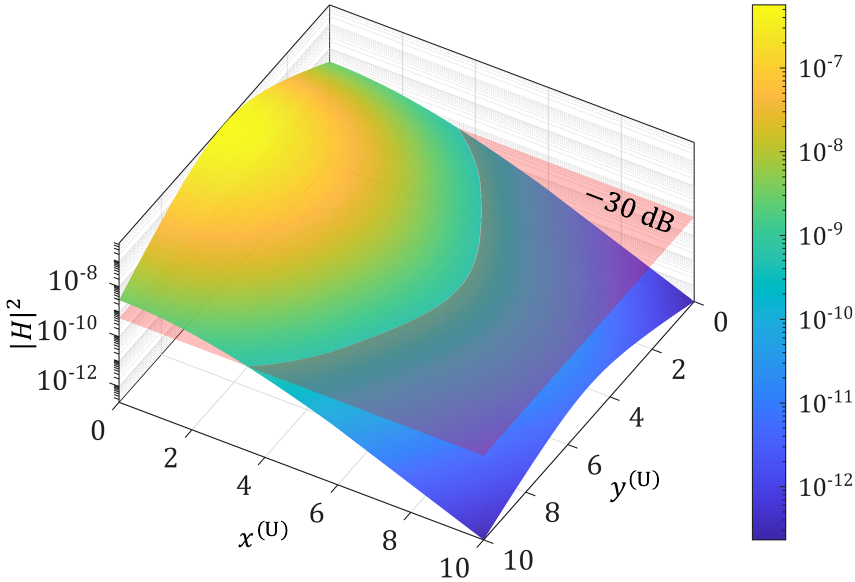}
        \caption{}
    \end{subfigure}
\caption{The spatial distributions of $y^*$ and $|H|^2$ under the single-waveguide and single-antenna configuration.}
\label{f:map_y_star}
\end{figure}

\begin{figure*}[!tb]
\centering
    \begin{subfigure}[!tb]{0.245\linewidth}
        \centering
        \includegraphics[width=1\linewidth]{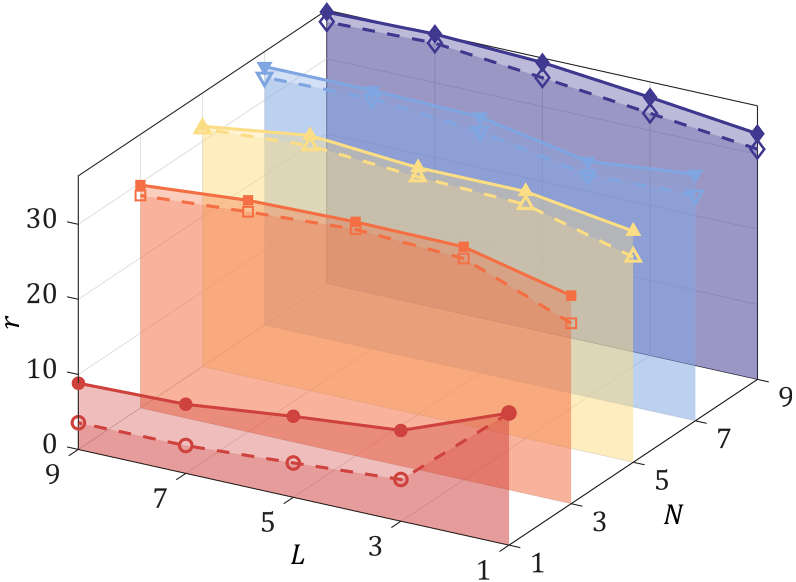}
        \caption{$M = 2$}
    \end{subfigure}
    \begin{subfigure}[!tb]{0.245\linewidth}
        \includegraphics[width=1\linewidth]{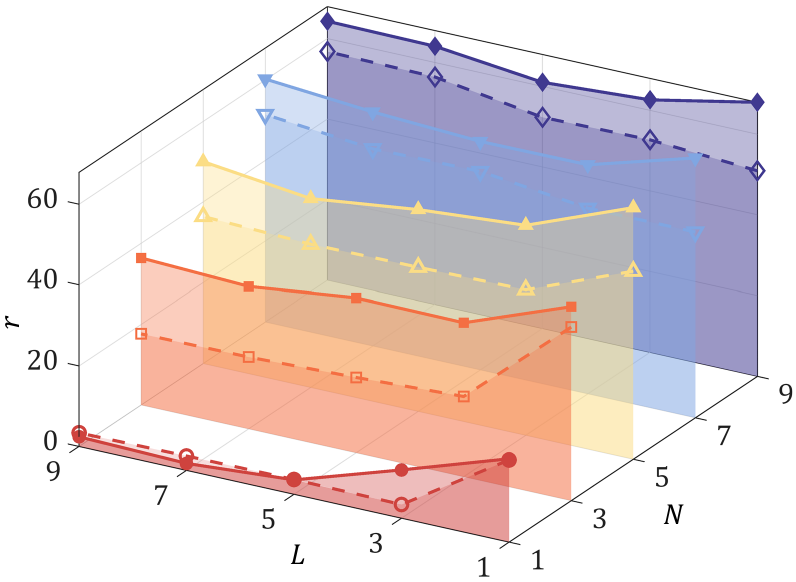}
        \caption{$M = 5$}
    \end{subfigure}
    \begin{subfigure}[!tb]{0.245\linewidth}
        \includegraphics[width=1\linewidth]{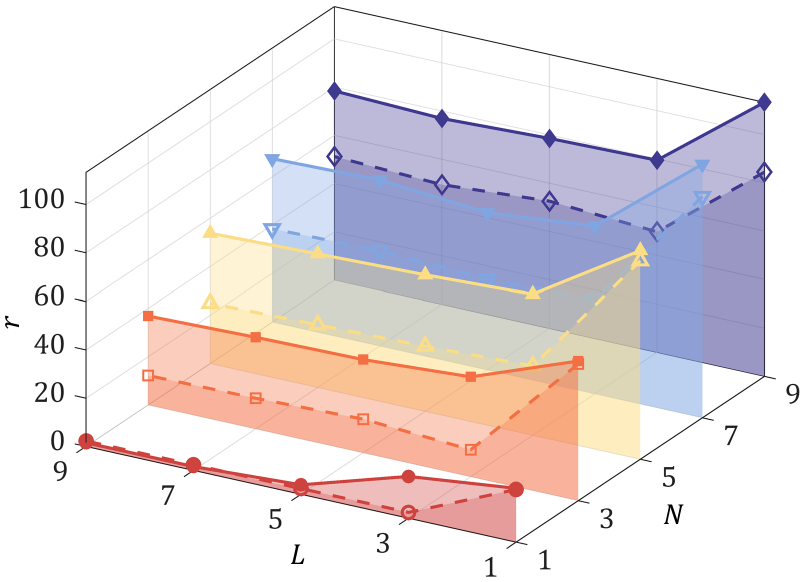}
        \caption{$M = 10$}
    \end{subfigure}
    \begin{subfigure}[!tb]{0.245\linewidth}
        \includegraphics[width=1\linewidth]{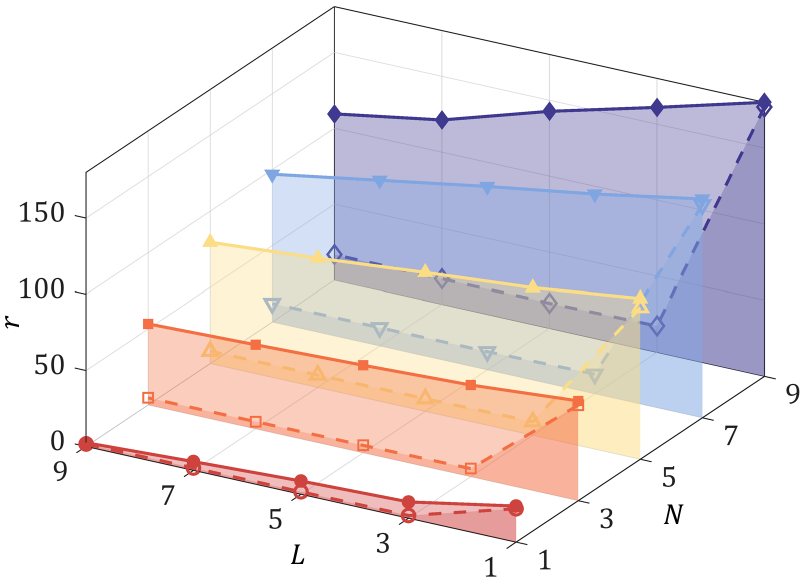}
        \caption{$M = 100$}
    \end{subfigure}
\caption{Comparison of the achievable sum rate under different $N$, $L$, and $M$, evaluated using two beamforming matrix design schemes: the solid curves correspond to WMMSE, while the dashed curves correspond to ZF.
}
\label{f:sumrate}
\end{figure*}

\begin{figure*}[!tb]
\centering
    \begin{subfigure}[!tb]{0.245\linewidth}
        \centering
        \includegraphics[width=1\linewidth]{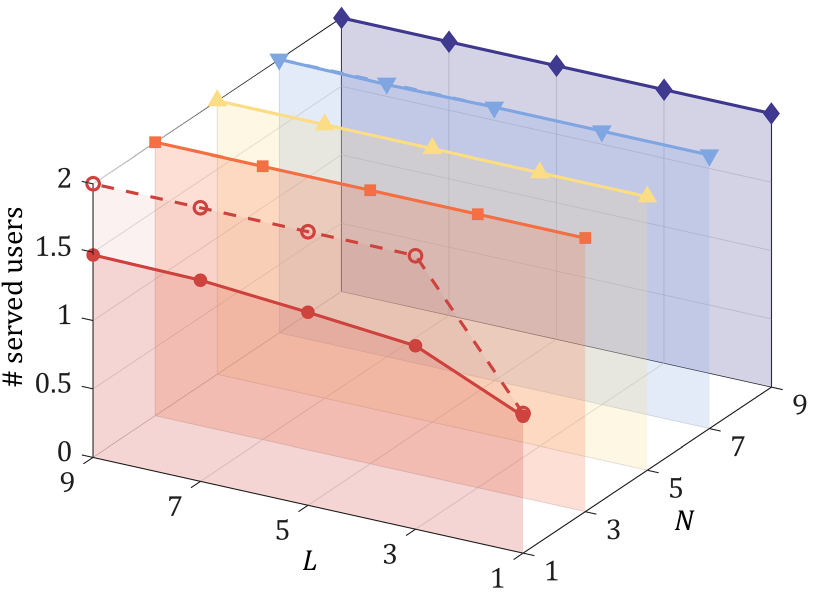}
        \caption{$M = 2$}
    \end{subfigure}
    \begin{subfigure}[!tb]{0.245\linewidth}
        \includegraphics[width=1\linewidth]{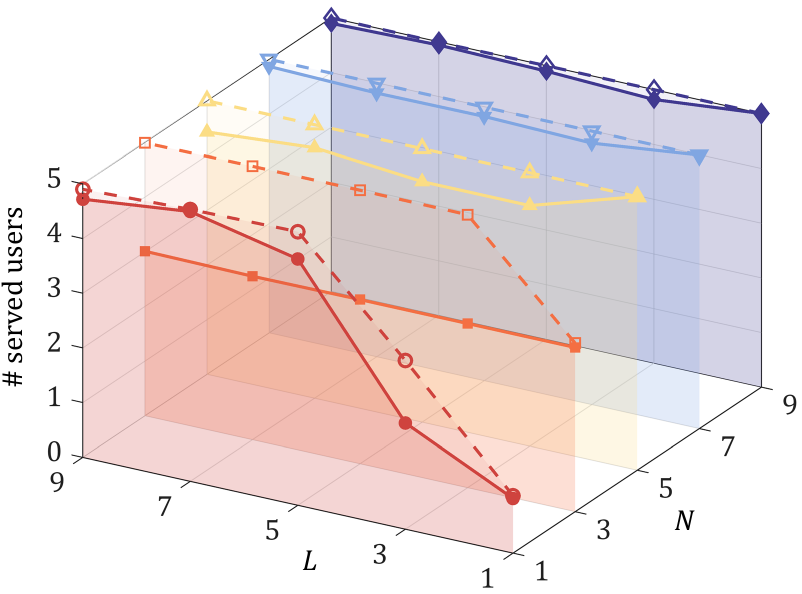}
        \caption{$M = 5$}
    \end{subfigure}
    \begin{subfigure}[!tb]{0.245\linewidth}
        \includegraphics[width=1\linewidth]{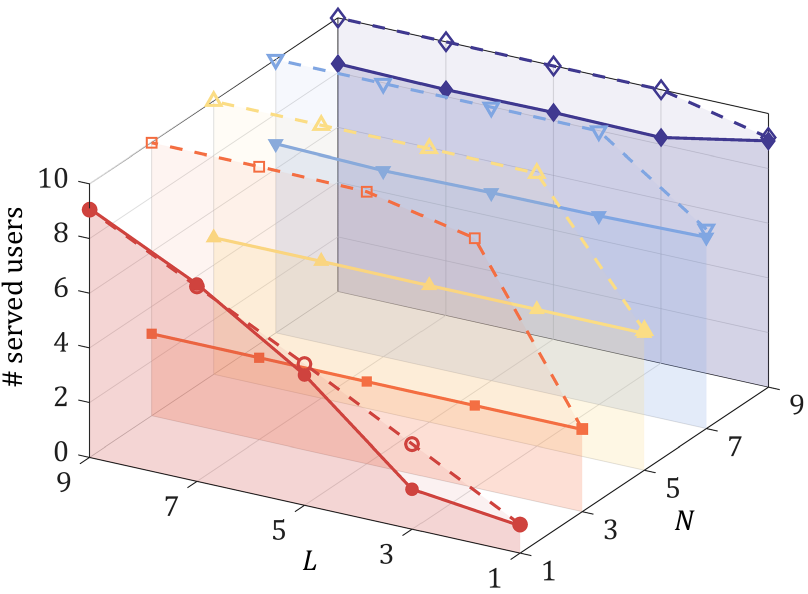}
        \caption{$M = 10$}
    \end{subfigure}
    \begin{subfigure}[!tb]{0.245\linewidth}
        \includegraphics[width=1\linewidth]{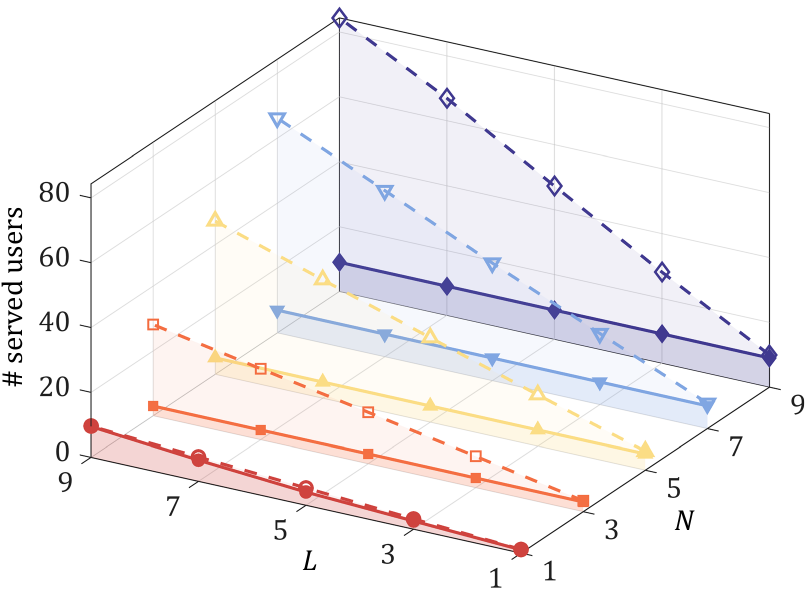}
        \caption{$M = 100$}
    \end{subfigure}
\caption{Comparison of the number of users successfully served under different values of $N$, $L$, and $M$ using two beamforming matrix design schemes: the solid curves correspond to WMMSE, while the dashed curves correspond to ZF.
}
\label{f:served_user}
\end{figure*}

\subsection{Multi-PA System Performance}
We now evaluate the multi-PA performance using the optimization framework developed in Section \ref{sec:IVB}, comparing WMMSE [29] and ZF [30] beamforming strategies across various system configurations. 

Fig.~\ref{f:sumrate} reveals several key trends in achievable sum rate under varying waveguide densities, antenna densities, and user densities. The observations are as follows:
\begin{itemize}[leftmargin=0.5cm]
    \item Increasing the number of waveguides $N$ consistently improves sum rate by providing more independent channels. This benefit is particularly dramatic when $N < M$ , where sum rate grows almost linearly with $N$, while saturation occurs when $N > M$ due to diminishing diversity returns.
    \item Increasing antennas per waveguide $L$ only benefits systems with excess waveguide diversity ($N > M$). When $N \leq M$, since multiple antennas on the same waveguide emit mutually coherent signals (differing only by a complex constant), increasing $L$, although allowing more users to be served, introduces additional interference and therefore does not enhance the overall sum rate.
    \item ZF approaches WMMSE performance only under favorable conditions (e.g., $N \geq M$ or $L=1$), where interference can be effectively nullified. In more challenging scenarios, WMMSE significantly outperforms ZF by optimally balancing signal enhancement and interference suppression.
\end{itemize}

Furthermore, we investigate the number of users that can be effectively served under different system configurations, as shown in Fig.~\ref{f:served_user}. The number of served users is defined as the number of users with non-zero data rates. It can be observed that, for $N > M$, the WMMSE-based system is capable of serving all users simultaneously. When $N \leq M$, however, WMMSE maximizes the sum rate by automatically allocating more power to users with favorable channel conditions, while suppressing the rates of weak-channel users to nearly zero. In contrast, the ZF scheme prioritizes serving as many users as possible, albeit at the expense of a lower overall sum rate.

Between these two beamforming (BF) design strategies, WMMSE achieves a higher sum rate, whereas ZF provides broader user coverage. As can be observed in Fig.~\ref{f:sumrate}, when a single antenna is deployed per waveguide, the ZF scheme can serve at most $N$ users without introducing additional interference. Therefore, the choice between WMMSE and ZF can be flexibly made depending on the system requirements and deployment scenarios.

\section{Conclusion}
This paper has undertaken a critical shift in the paradigm for modeling and optimizing PASS. By moving beyond the prevalent yet physically inaccurate omni-directional assumption, we have established the DiPASS framework, which places the inherent directional radiation of leaky-wave antennas at the forefront of system design.

Our core contribution is a comprehensive and physically consistent channel model that integrates three pivotal real-world factors: a Gaussian beam pattern for directional radiation, a practical waveguide attenuation coefficient, and a stochastic LoS blockage model. A key enabler within this model is the novel ``equal quota division'' strategy, which resolves a fundamental practical barrier by ensuring predetermined, position-independent coupling lengths for multiple antennas on a single waveguide.



While this work has focused on a static channel model, the DiPASS framework opens several promising avenues for future research. The dynamic optimization of PASS in time-varying channels with mobile users presents a significant challenge. Furthermore, integrating DiPASS with other emerging 6G technologies, such as integrated sensing and communication and joint communication and computation, could unlock new functionalities. Finally, the exploration of more complex waveguide topologies and multi-band operation remains an open and exciting area.

In summary, this work transitions PASS from a proven concept to a physically grounded and practically deployable architecture. By embracing the directional nature of its core components, DiPASS provides a realistic performance benchmark and a powerful set of design principles, paving the way for its successful adoption in future 6G and terahertz communication networks.

\appendices

\section{Proof of Theorem~\ref{thm:condition}}\label{sec:AppA}
To prove Theorem \ref{thm:condition}, we have to analyze the trade-off between waveguide loss and wireless propagation reliability in \eqref{e:partial_ln_H_2}.
To this end, we define a new function
\begin{equation*}
f_{y^{(\text{P})}}(\gamma) \triangleq -\alpha_{\text{W}} - 2 \ln \alpha_{\text{L}} \cdot \frac{\gamma}{\sqrt{A + \gamma^2}} + \frac{2\gamma}{A + \gamma^2},
\end{equation*}
where $\gamma \triangleq y^{(\text{U})} - y^{(\text{P})}$ and $A \triangleq (x^{(\text{U})} - x^{(\text{P})})^2 + (z^{(\text{U})} - z^{(\text{P})})^2$.

The derivative of $f_{y^{(\text{P})}}$ with respect to $\gamma$ is
\begin{equation*}
\frac{\partial f_{y^{(\text{P})}}(\gamma)}{\partial \gamma} = \frac{-2 \ln \alpha_{\text{L}} \cdot A \sqrt{A + \gamma^2} + 2(A - \gamma^2)}{(A + \gamma^2)^2}.
\end{equation*}
The function $f_{y^{(\text{P})}}(\gamma)$ reaches its extremum points when $\frac{\partial f_{y^{(\text{P})}}(\gamma)}{\partial \gamma} = 0$. Since the denominator of $\frac{\partial f_{y^{(\text{P})}}(\gamma)}{\partial \gamma}$ is always positive, we solve
\begin{equation*}
\gamma^2 = \frac{2A + (\ln \alpha_{\text{L}})^2 A^2 \pm \ln \alpha_{\text{L}}  A \sqrt{(\ln \alpha_{\text{L}})^2 A^2 + 8A}}{2}.
\end{equation*}

Given that $\ln \alpha_{\text{L}} < 0$ and $\gamma^2 \geq 0$, we can discard the smaller root. After some mathematical manipulations, we obtain two extremum points of $f_{y^{(\text{P})}}(\gamma)$ as follows:
\begin{align*}
\gamma_{\text{max}} \!&=\! \sqrt{\frac{(2A \!+\! (\ln \alpha_{\text{L}})^2 A^2) \!-\! \ln \alpha_{\text{L}} \!\cdot\! A \sqrt{(\ln \alpha_{\text{L}})^2 A^2 \!+\! 8A}}{2}}, \\
\gamma_{\text{min}} \!&=\! -\sqrt{\frac{(2A \!+\! (\ln \alpha_{\text{L}})^2 A^2) \!-\! \ln \alpha_{\text{L}} \!\cdot\! A \sqrt{(\ln \alpha_{\text{L}})^2 A^2 \!+\! 8A}}{2}}.
\end{align*}
Observing $f_{y^{(\text{P})}}(\gamma)$, for $\gamma_{\text{min}} < 0$, we have $f_{y^{(\text{P})}}(\gamma_{\text{min}}) < 0$. However, if $f_{y^{(\text{P})}}(\gamma_{\text{max}}) < 0$, $\frac{\partial \ln |H|^2}{\partial y^{(\text{P})}} < 0$ always holds, meaning that \eqref{e:partial_ln_H_2} can only achieve its maximum at $y^{(\text{P})} = 0$.


Therefore, a non-origin optimal solution exists only when $f(\gamma_{\text{max}}) > 0$. Substituting $\gamma_{\text{max}}$ into the expression for $f_{y^{(\text{P})}}(\gamma)$, we have
\begin{eqnarray}\label{e:f_yP}
 &&\hspace{-1cm} f_{y^{(\text{P})}}(\gamma_{\text{max}}) =  -\alpha_{\text{W}} - 2 \ln \alpha_{\text{L}} \notag\\
 &&\hspace{0cm} \cdot\sqrt{\frac{2A + (\ln \alpha_{\text{L}})^2 A^2 - \ln \alpha_{\text{L}} \cdot A \sqrt{(\ln \alpha_{\text{L}})^2 A^2 + 8A}}{4A + (\ln \alpha_{\text{L}})^2 A^2 - \ln \alpha_{\text{L}} \cdot A \sqrt{(\ln \alpha_{\text{L}})^2 A^2 + 8A}}}\notag \\
 &&\hspace{1cm}+ 2 \frac{\sqrt{\frac{(2A + (\ln \alpha_{\text{L}})^2 A^2) - \ln \alpha_{\text{L}} \cdot A \sqrt{(\ln \alpha_{\text{L}})^2 A^2 + 8A}}{2}}}{\frac{4A + (\ln \alpha_{\text{L}})^2 A^2 - \ln \alpha_{\text{L}} \cdot A \sqrt{(\ln \alpha_{\text{L}})^2 A^2 + 8A}}{2}} \notag\\
 &&\hspace{0cm} \overset{(a)}{=}  -\alpha_{\text{W}} - 2 \ln \alpha_{\text{L}} \sqrt{\frac{T}{2A + T}} + 2\sqrt{2} \frac{\sqrt{T}}{2A + T},
\end{eqnarray}
where we have defined $2A + (\ln \alpha_{\text{L}})^2 A^2 - \ln \alpha_{\text{L}} \cdot A \sqrt{(\ln \alpha_{\text{L}})^2 A^2 + 8A} \triangleq T$ in step (a).

Taking the derivative of $T$ with respect to $\alpha_{\text{L}}$:
\begin{equation*}
\begin{aligned}
\frac{\partial T}{\partial \alpha_{\text{L}}} 
= &~ 2A^2 \ln \alpha_{\text{L}} \cdot \frac{1}{\alpha_{\text{L}}} - \frac{1}{\alpha_{\text{L}}} A \sqrt{(\ln \alpha_{\text{L}})^2 A^2 + 8A} \\
& - \left((\ln \alpha_{\text{L}})^2 A^2 + 8A\right)^{-1/2} (\ln \alpha_{\text{L}})^2 A^3 \cdot \frac{1}{\alpha_{\text{L}}}.
\end{aligned}
\end{equation*}
Since all terms are negative, $\frac{\partial T}{\partial \alpha_{\text{L}}} < 0$, indicating that $T$ is monotonically decreasing with respect to $\alpha_{\text{L}}$. Since $T(\alpha_{\text{L}} = 1) = 2A$, for $\alpha_{\text{L}} \in (0,1)$, we always have $T > 2A$.

Let $f_{\alpha_{\text{L}}}(\alpha_{\text{L}}) = f_{y^{(\text{P})}}(\gamma_{\text{max}})$. We now examine the extremum of $f_{\alpha_{\text{L}}}(\alpha_{\text{L}})$ with respect to $\alpha_{\text{L}}$:
\begin{equation*}
\begin{aligned}
&\frac{\partial f_{\alpha_{\text{L}}}(\alpha_{\text{L}})}{\partial \alpha_{\text{L}}} = 
-\frac{2}{\alpha_{\text{L}}} \sqrt{\frac{T}{2A + T}} \\
&\hspace{1cm}- \frac{2A \ln \alpha_{\text{L}} (2A + T)^{1/2} - \sqrt{2}(2A - T)}{\sqrt{T}(2A + T)^2} \cdot \frac{\partial T}{\partial \alpha_{\text{L}}}.
\end{aligned}
\end{equation*}
Substituting the specific expression for $T$ reveals that $2A \ln \alpha_{\text{L}} (2A + T)^{1/2} - \sqrt{2}(2A - T) = 0$.
Thus, $\frac{\partial f_{\alpha_{\text{L}}}(\alpha_{\text{L}})}{\partial \alpha_{\text{L}}} < 0$, indicating that $f_{\alpha_{\text{L}}}(\alpha_{\text{L}})$ is monotonically decreasing with respect to $\alpha_{\text{L}}$. Letting $T = uA$ and taking a specific value $\alpha_{\text{L}} = \exp\left(\frac{2 - u}{\sqrt{4 + 2u} \sqrt{A}}\right)$ (for $u > 2$), we have:
\begin{equation*} 
f_{y^{(\text{P})}}(\gamma_{\text{max}}) 
= -\alpha_{\text{W}} + \frac{u\sqrt{2u}}{2 + u} \cdot \frac{1}{\sqrt{A}}.
\end{equation*}
It can be easily shown that $f_u(u) = \frac{u\sqrt{2u}}{2 + u}$ is an increasing function for $u > -2$, and for specific $\alpha_{\text{W}}$ and $A$, there exists a feasible $u$ such that $f(\gamma_{\text{max}}) < 0$. For example, for $A > 9$ and $\alpha_{\text{W}} = 1.3$, when $u = 30$,
\begin{equation*}
f(\gamma_{\text{max}}) = -2W + \frac{15\sqrt{15}}{8} \cdot \frac{1}{\sqrt{A}} < 0.
\end{equation*}
Thus, when $\alpha_{\text{L}} > \left.\exp\left(\frac{2 - u}{\sqrt{4 + 2u} \sqrt{A}}\right)\right|_{u=30}  = \exp\left(\frac{-7}{2\sqrt{A}}\right)$, the optimal solution remains at the origin. Therefore, at least $\alpha_{\text{L}} < \exp\left(\frac{-7}{2\sqrt{A}}\right)$ is required, which implies that $T > uA = 30A \gg 2A$. Consequently, we can approximate \eqref{e:f_yP} as
\begin{eqnarray*}
    f_{y^{(\text{P})}}(\gamma_{\text{max}})
\hspace{-0.2cm}&\approx &\hspace{-0.2cm} -\alpha_{\text{W}} - 2 \ln \alpha_{\text{L}} + 2\sqrt{2} \frac{\sqrt{T}}{T} \notag\\
\hspace{-0.2cm}&\overset{(a)}{\approx} &\hspace{-0.2cm} -\alpha_{\text{W}} - 2 \ln \alpha_{\text{L}},
\end{eqnarray*}
where (a) follows from $\frac{1}{T} < \frac{1}{30A} \ll \alpha_{\text{W}}$, for typical values $A > 9$ and $\alpha_{\text{W}} = 1.3$.

Overall, there exists a non-origin optimal PA position only when $f(\gamma_{\text{max}}) > 0$. That is,
\begin{equation*}
-\alpha_{\text{W}} - 2 \ln \alpha_{\text{L}} > 0,
\end{equation*}
which simplifies to \eqref{eq:THMcondition}.

\section{Proof of Theorem~\ref{thm:optimal_y}}\label{sec:AppB}
Following \eqref{e:f_yP} in the proof of Theorem \eqref{thm:condition}, we introduce an intermediate variable by setting $\gamma = \frac{\sqrt{A} \cos \vartheta}{\sin \vartheta} > 0$ for $\vartheta \in (0, \pi/2)$. Then $\gamma^2 = \frac{A \cos^2 \vartheta}{\sin^2 \vartheta}$. Substituting it into $f_{y^{(\text{P})}}$, we have
\begin{equation}\label{e:f_yP_theta}
\begin{aligned}
f_{y^{(\text{P})}}(\gamma) 
&= -\alpha_{\text{W}} - 2 \ln \alpha_{\text{L}} \cdot \cos \vartheta + \frac{1}{\sqrt{A}} \sin(2\vartheta)\\
& \overset{(a)}{\approx}  -\alpha_{\text{W}} - 2 \ln \alpha_{\text{L}} \left(1 - \frac{1}{2} \vartheta^2\right) + \frac{1}{\sqrt{A}} \cdot 2\vartheta,
\end{aligned}
\end{equation}
where (a) follows from the second-order Taylor approximation.

The optimal position is achieved when $f_{y^{(\text{P})}}(\gamma) = 0$, leading to
\begin{equation*}
-\alpha_{\text{W}} - 2 \ln \alpha_{\text{L}} \left(1 - \frac{1}{2} \vartheta^2\right) + \frac{1}{\sqrt{A}} \cdot 2\vartheta = 0.
\end{equation*}
Since $\vartheta > 0$, we solve for $\vartheta^*$:
\begin{equation*}
\vartheta^* = \frac{-1 - \sqrt{A + \ln \alpha_{\text{L}}  A^2 (\alpha_{\text{W}} + 2 \ln \alpha_{\text{L}})}}{\ln \alpha_{\text{L}} \cdot A}.
\end{equation*}
Given that $\gamma^* = \frac{\sqrt{A} \cos \vartheta^*}{\sin \vartheta^*}$ and $\gamma \triangleq y^{(\text{U})} - y^{(\text{P})}$, we obtain the optimal PA position in \eqref{e:optimal_y}. Note that when $\gamma^* > y^{(\text{U})}$, we can only take $y^* = \max\{0,\, y^{(\text{U})}- \gamma^*\}$.

Next, we provide a more intuitive approximated solution to analyze the trend of optimal position variation.
Considering the second and third terms in \eqref{e:f_yP_theta}, we observe that:
\begin{equation*}
\frac{-2 \ln \alpha_{\text{L}} \cos \vartheta}{\frac{2}{\sqrt{A}} \cos \vartheta \sin \vartheta} > \frac{\sqrt{A} \alpha_{\text{W}}}{2 \sin \vartheta} > \frac{\sqrt{A} \alpha_{\text{W}}}{2} > 1.
\end{equation*}
Therefore, we can reasonably neglect the third term and write
\begin{equation*}
f_{y^{(\text{P})}}(x) \approx -\alpha_{\text{W}} - 2 \ln \alpha_{\text{L}} \cos \vartheta.
\end{equation*}

Setting $f_{y^{(\text{P})}}(x) =0$, the optimal $\vartheta^*$ can be obtained, yielding the approximated $\gamma^*$ in \eqref{e:optimal_y_approx}.

\bibliographystyle{IEEEtran}
\bibliography{Ref}

\end{document}